%% file: main.tex
\documentclass[copyright, creativecommons]{eptcs}
\hypersetup{hidelinks}

\usepackage{amsthm}
\usepackage{amsmath}
\usepackage{amssymb}
\usepackage{enumerate}
\usepackage{graphicx}
\usepackage{iftex}
\usepackage[svgnames,table]{xcolor}
\usepackage[normalem]{ulem}

\ifpdf
  \usepackage{underscore}         
  \usepackage[T1]{fontenc}        
\else
  \usepackage{breakurl}           
\fi

\usepackage{tikz}
\IfFileExists{SkipFigure}
{\usetikzlibrary{calc,angles,arrows}}
{\usetikzlibrary{calc,angles,quotes,arrows}}

\title{A Cancellation Law for Probabilistic Processes}

\author{%
  Rob van Glabbeek%
  \thanks{Supported by Royal Society Wolfson Fellowship
      RSWF{\textbackslash}R1{\textbackslash}221008} 
  \institute{University of Edinburgh}
  \institute{University of New South Wales}
  \email{rvg@stanford.edu}
\and
  Jan Friso Groote
  \institute{Eindhoven University of Technology}
  \email{j.f.groote@tue.nl}
\and
  Erik de Vink
  \institute{Eindhoven University of Technology}
  \email{evink@win.tue.nl}
}

\def\titlerunning{A Cancellation Law for Probabilistic Processes}
\def\authorrunning{Van Glabbeek, Groote \& De Vink}

\input{commands}

\begin{document}
\maketitle

\input{intro}
\input{preliminaries}

\input{processes}

\input{properties}
\input{continuity}

\input{cancellativity}

\input{conclusion}

\bibliographystyle{eptcs}
\bibliography{main}

\end{document}

%% file: commands.tex
\makeatletter
\newcommand{\xRightarrow}[2][]{\ext@arrow
  0359\Rightarrowfill@{#1}{#2}}
\makeatother

\newcommand{\Arrow}{\xRightarrow{\, {} \; \,}}
\newcommand{\DistrE}{\Distr(\calE)}
\newcommand{\DistrF}{\Distr(\calF)}
\DeclareMathAlphabet{\mathbbm}{U}{bbm}{m}{n}          
\newcommand{\IN}{\mathbbm{N}}                         
\newcommand{\IR}{\mathbbm{R}}

\newcommand{\after}{\mathord{\upharpoonright}}
\newcommand{\alphaarrow}{\arrow{\alpha}}
\newcommand{\alphahidearrow}{\arrow{(\alpha)}}
\newcommand{\arrow}[1]{%
  \mkern1mu \xrightarrow{\raisebox{0ex}[0.1ex][-0.1ex]{\scriptsize $\,#1\,$}}}
\newcommand{\astop}{a \pref \partial(\bfzero)}
\newcommand{\bfzero}{\normalfont{\textbf{0}}}
\newcommand{\bigoplusiinI}{\textstyle{\bigoplus_{i \mathord{\in} I}}}
\newcommand{\bigoplusjinJ}{\textstyle{\bigoplus_{j \mathord{\in} J}}}
\newcommand{\bigoplusjinJi}{\textstyle{\bigoplus_{j \mathord{\in} J_i}}}

\newcommand{\bigopluskinK}{\textstyle{\bigoplus_{k \mathord{\in} K}}}
\newcommand{\bisim}{\mathrel{\,%
  \raisebox{.3ex}{$\underline{\makebox[.7em]{$\leftrightarrow$}}$}\,}}
\newcommand{\bisimilarity}{bisimilarity}
\newcommand{\bnfeq}{\mathrel{{:}{:}{=}}}
\newcommand{\brbisim}{\mathrel{\,%
  \raisebox{.3ex}{$\underline{\makebox[.8em]{$\leftrightarrow$}}_{\,
      b}$}\mkern-1mu}}
\newcommand{\bstop}{b \pref \partial(\bfzero)}

\newcommand{\cstop}{c \pref \partial(\bfzero)}
\newcommand{\den}[1]{\mbox{$[\hspace{-1.6pt}[$}#1\mbox{$]\hspace{-1.6pt}]$}}
\newcommand{\half}{\textstyle{\frac12}}

\newcommand{\lc}{\lbrace \:}
\newcommand{\mopcalR}{\,{\calR}\,}
\newcommand{\mydot}{{\cdot}}
\newcommand{\mycdot}{\mathop{\cdot}}
\newcommand{\myeq}{{} = \:}

\newcommand{\mysmallfrac}[2]{\frac{%
    \mbox{\small $#1 \rule[-3pt]{0pt}{7pt}$}}{
    \mbox{\small $#2 \rule{0pt}{7pt}$}}}
\newcommand{\mytinyfrac}[2]{\mbox{\tiny $\textstyle{\frac{#1}{#2}}$}}

\newcommand{\oplusr}{\probc{r}}

\newcommand{\pref}{\mathop{.}}
\newcommand{\probc}[1]{\mathbin{\mbox{${}_{\scriptstyle #1
  \mkern1mu} \hspace{-.8pt}\oplus\,$}}} 

\newcommand{\rc}{\: \rbrace}
\newcommand{\singleton}[1]{\lbrace {#1} \rbrace}
\newcommand{\sosrule}[2]{%
  \def\arraystretch{0.50}
  \begin{array}{c} {#1}\rule[-4pt]{0pt}{11pt} \\ \hline
    \rule{0pt}{13pt}{#2} \end{array}
  \def\arraystretch{1.0}}
\newcommand{\stardot}{\mycdot}

\newcommand{\sumiinI}{{\textstyle\sum_{\, i {\in} I}}}
\newcommand{\sumjinJ}{{\textstyle\sum_{\, j {\in} J}}}
\newcommand{\tauarrow}{\arrow{\tau}}
\newcommand{\tauhidearrow}{\arrow{(\tau)}}
\newcommand{\ttfrac}[2]{\textstyle{\frac{#1}{#2}}}

\newcommand{\etabar}{\bar{\eta}}
\newcommand{\mubar}{\bar{\mu}}
\newcommand{\nubar}{\bar{\nu}}
\newcommand{\muhat}{\hat{\mu}}
\newcommand{\nuhat}{\hat{\nu}}

\newcommand{\Distr}{\mathit{Distr}}

\newcommand{\spt}{\mathit{spt}}
\newcommand{\wgt}{\mathit{wgt}}

\newcommand{\calA}{\mathcal{A}}
\newcommand{\calE}{\mathcal{E}}
\newcommand{\calF}{\mathcal{F}}
\newcommand{\calP}{\mathcal{P}}
\newcommand{\calR}{\mathcal{R}}

\newtheorem{theorem}{Theorem}
\newtheorem{definition}[theorem]{Definition}
\newtheorem{lemma}[theorem]{Lemma}
\newtheorem{corollary}[theorem]{Corollary}

\newcommand{\blankline}{\vspace*{\baselineskip}}

%% file: intro.tex

\begin{abstract}
  We show a cancellation property for probabilistic
  choice. If $\mu \oplus \varrho$ and $\nu \oplus \varrho$ are branching
  probabilistic  bisimilar, then $\mu$
  and~$\nu$ are also branching probabilistic bisimilar. We do this in 
  the setting of a basic process language
  involving non-deterministic and probabilistic choice and define branching
  probabilistic bisimilarity on distributions.
  Despite the fact that the cancellation property is very elegant and
  concise, we failed to provide a short and natural combinatorial
  proof. Instead we provide a proof using metric topology. Our major
  lemma is that every distribution can be unfolded into an equivalent
  stable distribution, where the topological arguments are required to
  deal with uncountable branching.
\end{abstract}

\section{Introduction}
\label{sec:intro}

A familiar property of the real numbers~$\IR$ is the additive
cancellation law: if $x + z = y + z$ then~$x = y$.
Switching to the Boolean setting, and interpreting~$+$ by~$\lor$ and
$=$ by~$\Leftrightarrow$, the property becomes
$(x \lor z) \Leftrightarrow (y \lor z)$ implies $x \Leftrightarrow y$.
This is not generally valid. Namely, if $z$~is true, nothing can be derived
regarding the truth values of $x$ and~$y$.
Algebraically speaking, the reals provide an `additive inverse', and the
Booleans do not have a `disjunctive' version of it.

A similar situation holds for strong bisimilarity in the pure
non-deterministic setting vs.\ strong bisimilarity in the mixed
non-deterministic and probabilistic setting.
When we have $E + G \bisim F + G$ for the non-deterministic processes
$E + G$ and~$F + G$, it may or may not be the case that~$E \bisim F$.
However, if
${P \probc{1/2 \mkern1mu} R} \, \bisim \, {Q \probc{1/2 \mkern1mu} R}$
for the probabilistic processes $P \probc{1/2 \mkern1mu} R$ and
$Q \probc{1/2 \mkern1mu} R$, with probabilistic
choice~$\probc{1/2 \mkern1mu}$, we can exploit a semantic
characterization of bisimilarity as starting point of a
calculation. The characterization reads
\begin{equation}
  P \bisim Q
  \quad \text{iff} \quad
  \forall C \in \calE/{\bisim} \colon
  \mu[C] = \nu[C]
  \label{eqn-characterization}
\end{equation}
where the distributions $\mu, \nu \in \DistrE$ are induced by $P$
and~$Q$.
To spell out the above, two probabilistic processes $P$ and~$Q$ are
strongly bisimilar iff the distributions $\mu$ and~$\nu$ induced by
$P$ and~$Q$, respectively, assign the same probability to every
equivalence class~$C$ of non-deterministic processes modulo strong
bisimilarity.
In the situation that
$P \probc{1/2 \mkern1mu} R \bisim Q \probc{1/2 \mkern1mu} R$ we obtain
from~(\ref{eqn-characterization}), for equivalence classes
$C \in \calE/{\bisim}$ and distributions $\mu$, $\nu$, and~$\rho$
induced by the processes $P$, $Q$, and~$R$, that
\begin{displaymath}
  \begin{array}{l}
    P \probc{1/2 \mkern1mu} R \,\bisim\, Q \probc{1/2 \mkern1mu} R \implies
    \forall C \in \calE/{\bisim} \colon
    \half \mu[C] + \half \varrho[C] = 
    \half \nu[C] + \half \varrho[C] \implies
    {} \smallskip \\ \qquad \qquad
    \forall C  \in \calE/{\bisim} \colon
    \half \mu[C] = \half \nu[C] \implies
    \forall C  \in \calE/{\bisim} \colon
    \mu[C] = \nu[C] \implies
    P \bisim Q
  \end{array}
\end{displaymath}
relying on the arithmetic of the reals.

We are interested in whether the cancellation law 
also holds for weaker notions of process equivalence for probabilistic
processes, especially for branching probabilistic bisimilarity as proposed
in~\cite{GGV19:fc}. We find that it does but the proof is involved. 
A number of initial attempts were directed towards finding a straightforward 
combinatorial proof, but all failed. A proof in
a topological setting, employing the notion of sequential compactness to deal 
with potentially infinite sequences of transitions is reported in this paper. 
We leave the existence of a shorter, combinatorial proof as an open question.
 
Our strategy to prove the above cancellation law
for probabilistic processes and branching probabilistic bisimilarity
is based on two intermediate results: (i)~every probabilistic process
unfolds into a so-called \emph{stable} probabilistic process, and
(ii)~for stable probabilistic processes a characterization of the
form~(\ref{eqn-characterization}) does hold.
Intuitively, a stable process
is a process that cannot do an internal move without leaving its
equivalence class.

In order to make the above more concrete, let us consider an example.
For the ease of presentation we use distributions directly, rather
than probabilistic processes.
Let the distributions $\mu$ and~$\nu$ be given by\vspace{-2ex}
\begin{align*}
  \mu & \myeq \half \delta(\astop) \oplus \half \delta(\bstop)
  \smallskip \\
  \nu & \myeq \ttfrac13 \delta( \tau \pref \mkern1mu (
  \partial(\astop)
  \probc{\mytinyfrac12} 
  \partial(\bstop) )) \oplus \ttfrac13 \delta(\astop) \oplus \ttfrac13
  \delta(\bstop)
\end{align*}
with $a$ and~$b$ two different actions. The distribution~$\mu$
assigns probability $0.5$ to $a \pref \partial(\bfzero)$, meaning an
$a$-action followed by a deadlock with probability~$1$, and
probability $0.5$ to $b \pref \partial(\bfzero)$, i.e.\ a $b$-action
followed by deadlock with probability~$1$. The distribution~$\nu$
assigns both these non-deterministic processes probability~$\frac13$
and assigns the remaining probability~$\frac13$ to
$\tau \pref \mkern1mu ( \partial(\astop) \probc{\mytinyfrac12}
\partial(\bstop) )$, where a $\tau$-action precedes 
a 50-50 percent
choice between the processes mentioned earlier. Below, we show that
$\mu$ and~$\nu$ are branching probabilistic bisimilar,
i.e.~$\mu \brbisim \nu$. However, if $C_1$, $C_2$ and~$C_3$ are the
three different equivalence classes of
$\tau \pref \mkern1mu ( \partial(\astop) \probc{\mytinyfrac12}
\partial(\bstop)$, $\astop$ and~$\bstop$, respectively, we have
\begin{displaymath}
  \mu[C_1] = 0 \neq \ttfrac13 = \nu[C_1] ,\ 
 \mu[C_2] = \ttfrac12 \neq \ttfrac13 = \nu[C_2] ,\
 \text{and}\
 \mu[C_3] = \ttfrac12 \neq \ttfrac13 = \nu[C_3].
\end{displaymath}
Thus, although $\mu \brbisim \nu$, it does not hold that
$\mu[C] = \nu[C]$ for every equivalence class~$C$. Note that the
distribution~$\nu$ is not stable, in the sense that it allows an
internal transition to the branchingly equivalent~$\nu$.

As indicated, we establish in this paper a cancellation law for branching
probabilistic bisimilarity in the context of mixed non-deterministic
and probabilistic choice, exploiting the process language
of~\cite{BS01:icalp}, while dealing with distributions of finite
support over non-deterministic processes for its semantics. We propose
the notion of a stable distribution and show that every distribution
can be unfolded into a stable distribution by chasing its (partial)
$\tau$-transitions.  Our framework, including the notion of branching
probabilistic bisimulation, builds on that of~\cite{GV19:sg65,GGV19:fc}.

Another trait of the current paper, as in~\cite{GV19:sg65,GGV19:fc},
is that distributions are taken as semantic foundation for
bisimilarity, rather than seeing bisimilarity primarily as an
equivalence relation on non-deterministic processes, which is
subsequently lifted to an equivalence relation on distributions, as
is the case for the notion of branching probabilistic bisimilarity
of~\cite{SL94:concur,Seg95:thesis} and also
of~\cite{AW06:tcs,AGT12:tcs}. The idea to consider distributions as
first-class citizens for probabilistic bisimilarity stems
from~\cite{EHKTZ13:qest}. In the systematic overview of the
spectrum~\cite{BDH2020:acta}, also Baier et al.\ argue that a
behavioral relation on distributions is needed to properly deal with
silent moves. 

Metric spaces and complete metric spaces, as well as their associated
categories, have various uses in concurrency theory. In the setting of
semantics of probabilistic systems, metric topology has been advocated
as underlying denotational domain, for example
in~\cite{BK00:mscs,HVB00:entcs,Nor97:phd}. For quantitative comparison
of Markov systems, metrics and pseudo-metric have been proposed for a
quantitative notion of behavior equivalence, see
e.g.~\cite{DGJP99:concur,GJS90:ifip,BW05:tcs}. The specific use of
metric topology in this paper to derive an existential property of a
transition system seems new.

The remainder of the paper is organized as follows.
Section~\ref{sec-preliminaries} collects some definitions
from metric topology and establishes some auxiliary results.
A simple process language with non-deterministic and probabilistic
choice is introduced in Section~\ref{sec-process-language}, together
with examples and basic properties of the operational
semantics.
Our definition of branching probabilistic bisimilarity is given in
Section~\ref{sec-props}, followed by a congruence result with respect
to probabilistic composition and a confluence property.
The main contribution of the paper is presented in
Sections~\ref{sec-continuity} and \ref{sec-cancellativity}. Section 5
shows in a series of continuity lemmas that the set
of branching probabilistic bisimilar descendants is a (sequentially)
compact set. Section 6 exploits these results to argue that unfolding
of a distribution by inert $\tau$-transitions has a stable end point,
meaning that a stable branchingly equivalent distribution can be
reached. With that result in place, a cancellation law for branching
probabilistic bisimilarity is established.
Finally, Section~\ref{sec-conclusion} wraps up with concluding remarks
and a discussion of future work.

%% file: preliminaries.tex

\section{Preliminaries}
\label{sec-preliminaries}

For a non-empty set~$X$, we define $\Distr(X)$ as the set of all
probability distributions over~$X$ of finite support, i.e.,
$\Distr(X) = \lc \mu \colon X \to [0,1] \mid
\text{$\textstyle{\sum}_{\, x \mathbin\in X} \: \mu(x) = 1$,
  $\mu(x) > 0$ for finitely many~$x \in X$} \rc$.  We use $\spt(\mu)$
to denote the finite set $\lc x \in X \mid \mu(x) > 0 \rc$.
Often, we write $\mu = \bigoplusiinI \: p_i \stardot x_i$ for an index
set~$I$, $p_i \geqslant 0$ and $x_i \in X$ for~$i \in I$, where
$p_i > 0$ for finitely many~$i \in I$. Implicitly, we assume
$\sumiinI \: p_i = 1$.
We also write $r \mkern1mu \mu \oplus (1-r) \nu$ and, equivalently, 
$\mu \oplusr \nu$
for~$\mu, \nu \in \Distr(X)$ and $0 \leqslant r \leqslant 1$. As
expected, we have that
$(r \mkern1mu \mu \oplus (1-r) \nu)(x)=(\mu \oplusr \nu)(x) = r
\mkern1mu \mu(x) + (1-r) \nu(x)$ for~$x \in X$.
The \textit{Dirac distribution} on~$x$, the unique distribution
with support $x$, is denoted $\delta(x)$.

The set $\Distr(X)$ becomes a complete\footnote{A
    \emph{Cauchy sequence} is a sequence of points in a metric space
    whose elements become arbitrarily close to each other as the
    sequence progresses. The space is \emph{complete} if every such
    sequence has a limit within the space.}  metric space when
endowed with the sup-norm~\cite{Gir82:cata}, given by
$d(\mu,\nu) = \textstyle\sup_{x \mathbin\in X} \: | \mu(x) -
\nu(x)|$. This distance is also known as the
distance of uniform convergence or Chebyshev distance.

\begin{theorem}\label{thm-seq-compact}
If $Y \subseteq X$ is finite, then $\Distr(Y)$ is a sequentially
compact subspace of~$\Distr(X)$. This means that every sequence
in~$\Distr(Y)$ has a convergent subsequence with a limit in~$\Distr(Y)$.
\end{theorem}
\begin{proof}
    $\Distr(Y)$ is a bounded subset of $\IR^n$, where
    $n:=|Y|$ is the size of $Y$.  It also is closed. For $\IR^n$
    equipped with the Euclidean metric, the sequential compactness of
    closed and bounded subsets is known as the Bolzano-Weierstrass
    theorem~\cite{Lan97:utm}. When
    using the Chebyshev metric, the same proof applies.
\end{proof}

\noindent
In Section~\ref{sec-continuity} we use the topological structure
of the set of distributions over non-deterministic processes to study
unfolding of partial $\tau$-transitions. There we make use of the
following representation property.

\begin{lemma}
  \label{lem-representation-of-limit}
  Suppose the sequence of distributions $( \mu_i )_{i{=}0}^\infty$
  converges to the distribution~$\mu$ in~$\Distr(X)$. Then a sequence
  of distributions $( \mu'_i )_{i{=}0}^\infty$ in~$\Distr(X)$ and a
  sequence of probabilities $( r_i )_{i{=}0}^\infty$ in~$[0,1]$ exist
  such that
  $\mu_i = (1-r_i) \mkern1mu \mu \mathrel\oplus r_i \mkern1mu \mu'_i$
  for~$i \in \IN$ and $\lim_{\, i \rightarrow \infty} \: r_i = 0$.
\end{lemma}

\begin{proof}
  Let $i \in \IN$.  For $x \in \spt(\mu)$, the quotient
  ${\mu_i(x)}/{\mu(x)}$ is non-negative, but may exceed~$1$. However,
  $0 \leqslant \min \lc \mysmallfrac{\mu_i(x)}{\mu(x)} \mid x \in \spt(\mu)
  \rc \leqslant 1$, since the numerator cannot strictly exceed the
  denominator for all $x \in \spt(\mu)$. Let
  $r_i = 1 - \min \lc \mysmallfrac{\mu_i(x)}{\mu(x)} \mid x \in \spt(\mu)
  \rc$ for~$i \in \IN$. Then we have $0 \leqslant r_i \leqslant 1$.

  For~$i \in \IN$, define $\mu'_i \in \Distr(X)$ as follows. If
  $r_i > 0$ then
  $\mu_i'(x) = 1/{r_i} \cdot \bigl[ \mu_i(x) - (1-r_i) \mu(x) \bigr]$
  for~$x \in X$; if $r_i = 0$ then $\mu'_i = \mu$. We verify
  for~$r_i > 0$ that $\mu'_i$ is indeed a distribution: (i)~For
  $x \notin \spt(\mu)$ it holds that $\mu(x) = 0$, and therefore
  $\mu'_i(x) = {1}/{r_i} \cdot \mu_i(x) \geqslant 0$. For $x \in \spt(\mu)$,
  \begin{displaymath}
    \mu'_i(x) =
    {1}/{r_i} \cdot \bigl[ \mu_i(x) - (1-r_i) \mu(x) \bigr] =
    {\mu(x)} / {r_i} \cdot \bigl[ \mysmallfrac{\mu_i(x)}{\mu(x)} -
    \mysmallfrac{\mu_i(x_{\mathit{min}})}{\mu(x_{\mathit{min}})} \big] \geqslant 0
  \end{displaymath}
  for $x_{\mathit{min}} \in \spt(\mu)$ such that
  ${\mu_i(x_{\mathit{min}})}/{\mu(x_{\mathit{min}})}$ is
  minimal. (ii)~In addition,
  \begin{displaymath}
    \begin{array}{l}
      \textstyle{\sum} \lc \mu'_i(x) \mid x \in X \rc =
      {1}/{r_i} \cdot \textstyle{\sum} \lc \mu_i(x) \mid
      x \notin \spt(\mu) \rc +
      {1}/{r_i} \cdot \textstyle{\sum} \lc \mu_i(x) - (1-r_i) \mu(x)
      \mid x \in \spt(\mu) \rc =
      {} \smallskip \\ \qquad
      {1}/{r_i} \cdot \textstyle{\sum} \lc \mu_i(x) \mid
      x \in X \rc -
      {(1-r_i)}/{r_i} \cdot \textstyle{\sum} \lc \mu(x) \mid
      x \in \spt(\mu) \rc =
      {1}/{r_i} - {(1-r_i)}/{r_i} =
      {r_i}/{r_i} = 1.
    \end{array}
  \end{displaymath}
  Therefore, $0 \leqslant \mu'_i(x) \leqslant 1$ and $\sum \lc
  \mu'_i(x) \mid x \in X \rc = 1$.

  Now we prove that $\mu_i = (1-r_i) \mkern1mu \mu \mathrel\oplus r_i
  \mkern1mu \mu'_i$. 
  If~$r_i = 0$, then $\mu_i = \mu$, $\mu'_i = \mu$, and
  $\mu_i = (1-r_i) \mkern1mu \mu \oplus r_i \mkern1mu \mu'_i$.
  If~$r_i > 0$, then
  $\mu_i(x) = (1-r_i) \mkern1mu \mu(x) \oplus r_i \mkern1mu \mu'_i(x)$
  by definition of~$\mu'_i(x)$ for all~$x \in X$. Thus, also
  $\mu_i = (1-r_i) \mkern1mu \mu \oplus r_i \mkern1mu \mu'_i$ in this
  case.

  \newcommand{\xmn}{x'_{\mathit{min}}}

  Finally, we show that $\lim_{\, i \rightarrow \infty} \: r_i =
  0$. Let $\xmn \in \spt(\mu)$ be such that $\mu(\xmn)$ is
  minimal. Then we have
  \begin{displaymath}
    r_i =
    1 - \min \lc \mysmallfrac{\mu_i(x)}{\mu(x)} \mid
    x \in \spt(\mu) \rc =
    \max \lc \mysmallfrac{\mu(x) - \mu_i(x)}{\mu(x)} \mid
    x \in \spt(\mu) ,\, \mu(x) \geqslant \mu_i(x) \rc \leqslant
    \mysmallfrac{d(\mu,\mu_i)}{\mu(\xmn)}
  \end{displaymath}
  By assumption, $\lim_{\, i \rightarrow \infty} \: d(\mu,\mu_i) =
  0$. Hence also $\lim_{\, i \rightarrow \infty} \: r_i = 0$, as was to
  be shown.
\end{proof}

\noindent
The following combinatorial result is helpful in the sequel. 

\begin{lemma}
  \label{flexibelDelen}
  Let $I$ and~$J$ be finite index sets, $p_i,q_j\in[0,1]$ and
  $\mu_i,\nu_j \in \Distr(X)$, for $i \in I$ and $j \in J$, such that
  $\bigoplusiinI\: p_i\mu_i = \bigoplusjinJ \: q_j \nu_j$. Then
  $r_{ij} \geqslant 0$ and $\varrho_{ij} \in \Distr(X)$ exist such that
  $\sumjinJ \: r_{ij} = p_i$ and
  $p_i \stardot \mu_i = \bigoplusjinJ \: r_{ij} \stardot \varrho_{ij}$
  for all~$i \in I$, and $\sumiinI \: r_{ij} = q_j$ and
  $q_j \stardot \nu_j = \bigoplusiinI \: r_{ij} \stardot \varrho_{ij}$
  for all~$j \in J$.
\end{lemma}
\begin{proof}
  Let $\xi = \bigoplusiinI\: p_i\stardot\mu_i = \bigoplusjinJ \: q_j
  \stardot \nu_j$.
  We define
  $\textstyle r_{ij} = \sum_{x\in \mathit{spt}(\xi)} \:
  \displaystyle{\frac{\rule[-3pt]{0pt}{10pt} p_i \mkern1mu \mu_i(x)
      \cdot q_j \mkern1mu \nu_j(x)}{\rule{0pt}{7pt} \xi(x)}}$ for all
  $i \in I$ and $j \in J$.  In case $r_{ij} = 0$, choose
  $\varrho_{ij} \in \Distr(X)$ arbitrarily.  In case $r_{ij} \neq 0$,
  define $\varrho_{ij} \in \Distr(X)$, for $i \in I$ and $j \in J$, by
  \begin{displaymath}
    \varrho_{ij}(x) =
    \left\lbrace
      \begin{array}{cl}
        \displaystyle\frac
        {\rule[-3pt]{0pt}{10pt} p_i \mkern1mu \mu_i(x) \cdot
        q_j \mkern1mu \nu_j(x)}
        {\rule{0pt}{7pt} r_{ij} \, \xi(x)} \quad &
        \textrm{if } \xi(x)>0, \smallskip \\
        0 & \textrm{otherwise}
      \end{array}
    \right.
  \end{displaymath}
  for all $x \in X$. By definition of $r_{ij}$
  and~$\varrho_{ij}$ it holds that $\sum \lc \varrho_{ij}(x) \mid x
  \in X \rc = 1$. So, $\varrho_{ij} \in \Distr(X)$ indeed.

  We verify $\sumjinJ \: r_{ij} = p_i$ and
  $p_i \stardot \mu_i = \bigoplusjinJ \: r_{ij} \stardot \varrho_{ij}$ for~$i \in
  I$.
  \begin{align*}
    \sumjinJ \: r_{ij}
    & = \sumjinJ \: \textstyle \sum_{x \in \mathit{spt}(\xi)} \: p_i \mkern1mu
      \mu_i(x) \cdot q_j \mkern1mu \nu_j(x) / \xi(x) \\
    & = \textstyle \sum_{x \in \mathit{spt}(\xi)} \: p_i \mkern1mu
      \mu_i(x) \cdot \sumjinJ \: q_j \mkern1mu \nu_j(x) / \xi(x) \\
    & = \textstyle \sum_{x \in \mathit{spt}(\xi)} \: p_i \mkern1mu \mu_i(x)
    & (\text{since $\xi = \bigoplusjinJ \: q_j \stardot \nu_j$}) \\
    & = p_i \textstyle \sum_{x \in \mathit{spt}(\xi)} \: \mu_i(x) \\
    & = p_i \mkern1mu .
  \end{align*}
  Next, pick $y \in X$ and~$i \in I$. If $\xi(y) = 0$, then
  $p_i \mkern1mu \mu_i(y) = 0$, since
  $\xi(y) = \sumiinI \: p_i \mkern1mu \mu_i(y)$, and $r_{ij} = 0$ or
  $\varrho_{ij}(y) = 0$ for all~$j \in J$, by the various definitions,
  thus $\sumjinJ \: r_{ij} \mkern1mu \varrho_{ij}(y) = 0$ as well.

  Suppose $\xi(y) > 0$. Put $J_i = \lc j \in J \mid r_{ij} > 0 \rc$.
  If $j \in J \backslash J_i$, i.e.\ if $r_{ij} = 0$, then
  $p_i \mkern1mu \mu_i(y) q_j \mkern1mu \nu_j(y) / \xi(y) = 0$ by
  definition of~$r_{ij}$. Therefore we have
  \begin{align*}
    \sumjinJ \: r_{ij} \varrho_{ij}(y) \,
    & = \, \textstyle \sum_{j {\in} J_i} \: r_{ij} \mkern1mu \varrho_{ij}(y) \\
    & = \, \textstyle \sum_{j {\in} J_i} \: r_{ij} \mkern1mu p_i \mkern1mu 
      \mu_i(y) \cdot q_j \mkern1mu \nu_j(y) / (r_{ij} \mkern1mu \xi(y) ) \\
    & = \, \textstyle \sum_{j {\in} J_i} \: p_i \mkern1mu \mu_i(y) \cdot q_j
      \mkern1mu \nu_j(y) / \xi(y) \\
    & = \, \sumjinJ \: p_i \mkern1mu \mu_i(y) \cdot q_j \mkern1mu
      \nu_j(y) / \xi(y)
    & \text{(summand zero for $j \in J \backslash J_i$)} \\
    & = \, p_i \mkern1mu \mu_i(y) / \xi(y) \cdot
      \sumjinJ \: q_j \mkern1mu \nu_j(y)
    & \\
      & = \, p_i \mkern1mu \mu_i(y) \mkern1mu
    & (\text{since $\xi = \bigoplusjinJ \: q_j \stardot \nu_j$})
      \mkern1mu  .
  \end{align*}
  The statements $\sum_{i\in I} \: r_{ij}= q_j$ and
  $q_j \stardot \nu_j = \bigoplusiinI \: r_{ij} \stardot \varrho_{ij}$
  for~$j \in J$ follow by symmetry.
\end{proof}

%% file: processes.tex

\section{An elementary processes language}
\label{sec-process-language}

In this section we define a syntax and transition system semantics
for non-deterministic and probabilistic processes.
Depending on the top operator,
following~\cite{BS01:icalp}, a process is either a
non-deterministic process~$E \in \calE$, with constant~$\bfzero$,
prefix operators~$\alpha \pref {}$ and non-deterministic choice~$+$,
or a probabilistic process~$P \in \calP$, with the Dirac
operator~$\partial$ and probabilistic choices~$\oplusr$.

\begin{definition}[Syntax]
  \label{def-prob-processes}
  The classes $\calE$ and~$\calP$ of non-deterministic and
  probabilistic processes, respectively, over the set of
  actions~$\calA$, are given by
  \begin{displaymath}
    E \bnfeq \bfzero \mid \alpha \pref P \mid E + E
    \qquad \qquad
    P \bnfeq \partial(E) \mid P \oplusr P 
  \end{displaymath}
  with actions~$\alpha$ from~$\calA$ and where $0 \leqslant r
  \leqslant 1$.
\end{definition}

\noindent
We use $E, F, \ldots {}$ to range over~$\calE$ and $P, Q, \ldots {}$
to range over~$\calP$.
The probabilistic process $P_1 \mkern1mu \oplusr P_2$ behaves
as~$P_1$ with probability~$r$ and behaves as~$P_2$ with probability~$1-r$.

We introduce a complexity
measure~$c: \calE\cup \calP \to \IN$ for
non-deterministic and probabilistic processes based on the
size of a process. It is given by
$c(\bfzero) = 0$, $c(a \pref P) = c(P) + 1$, $c(E+F) = c(E)+c(F)$,
and~$c(\partial(E)) = c(E) + 1$, $c(P \oplusr Q) = c(P) + c(Q)$.

\paragraph{Examples}
As illustration, we provide the following pairs of 
non-deterministic processes, which are branching probabilistic bisimilar in 
the sense of Definition \ref{def-probabilistic-branching-bisimilar}. 
\begin{enumerate} [(i)]
\item $\mathbf{H_1} = a \pref \bigl( P \probc{\frac14} (
    P \probc{\frac13} Q ) \bigr)$ and $\mathbf{H_2} = a \pref \bigl( P
    \probc{\frac12} ( Q \probc{\frac12} Q ) \bigr)$
  \item $\mathbf{G_1} = a \pref ( P \probc{\frac12} Q )$ and
    $\mathbf{G_2} = a \pref \bigl( \partial \bigl( \tau \pref ( P
    \probc{\frac12} Q ) \bigl) \probc{\frac13} ( P \probc{\frac12} Q )
    \bigl)$
  \item $\mathbf{I_1} = a \pref \partial( b \pref P + \tau \pref Q )$
    and
    $\mathbf{I_2} = a \pref \partial( \tau \pref \partial( b \pref P +
    \tau \pref Q ) + b \pref P + \tau \pref Q)$
\end{enumerate}
The examples $\mathbf{H_1}$ and~$\mathbf{H_2}$ are taken
from~\cite{Hen12:facj}, and $\mathbf{G_1}$ and~$\mathbf{G_2}$ are taken
from~\cite{GGV19:fc}.
The processes $\mathbf{G_2}$ and~$\mathbf{I_2}$ contain a so-called inert
$\tau$-transition.

\blankline

\noindent
As usual, the SOS semantics for $\calE$ and~$\calP$ makes use of two
types of transition relations~\cite{HJ90:rtss,BS01:icalp,GGV19:fc}.

\pagebreak[3]

\begin{definition}[Operational semantics]
  \label{def-pr-transition-relation} \mbox{}
  \begin{enumerate}[(a)]
    \item 
    The transition relations ${\rightarrow} \subseteq
      \calE \times \calA \times \Distr(\calE)$ and ${\mapsto}
      \subseteq \calP \times \Distr(\calE)$ are given by\vspace{-2ex}
      \begin{displaymath}
        \begin{array}{c}
          \sosrule{P \mapsto \mu}{\alpha \pref P \arrow{\alpha} \mu}
          \: \textrm{\small\sc (pref)} \rule{0pt}{22pt}
          \medskip \\
          \sosrule{E_1 \arrow{\alpha} \mu_1}{E_1 + E_2 \arrow{\alpha}
            \mu_1}
          \: \textrm{\small\sc (nd-choice\,1)}
          \qquad
          \sosrule{E_2 \arrow{\alpha} \mu_2}{E_1 + E_2 \arrow{\alpha}
            \mu_2}
          \: \textrm{\small\sc (nd-choice\,2)}
          \bigskip \\
          \sosrule{}{\partial(E) \mapsto \delta(E)}
          \: \textrm{\small\sc (Dirac)}
          \qquad
          \sosrule{P_1 \mapsto \mu_1 \quad P_2 \mapsto
            \mu_2}{P_1 \oplusr P_2 \mapsto \mu_1 \oplusr \mu_2}
          \: \textrm{\small\sc (p-choice)}
        \end{array}
      \end{displaymath}
    \item 
      The transition relation ${\rightarrow}
      \subseteq \DistrE \times \calA \times \DistrE$ is such that $\mu
      \alphaarrow \mu'$ whenever $\mu = \bigoplusiinI \: p_i \stardot E_i$,
      $\mu' = \bigoplusiinI \: p_i \stardot \mu'_i $, and $E_i \alphaarrow
      \mu'_i$ for all $i \in I$.
      \vspace{-.5ex}
  \end{enumerate}
\end{definition}

\medskip

\noindent
In rule~\textsc{\small (Dirac)} of the relation ${\mapsto}$ we have
that the syntactic Dirac process~$\partial(E)$ is coupled to the
semantic Dirac distribution~$\delta(E)$. Similarly, in~\textsc{\small
  (p-choice)}, the syntactic probabilistic operator~$\oplusr$
in~$P_1 \oplusr P_2$ is replaced by semantic probabilistic composition
in $\mu_1 \oplusr \mu_2$. Thus, with each probabilistic
process~$P \in \calP$ we associate a
distribution~$\den{P} \in \DistrE$ as follows:
$\den{\partial(E)} = \delta(E)$ and
$\den{P \oplusr Q} = \den{P} \oplusr \den{Q}$, which is the
distribution $r \den{P} \oplus (1-r) \den{Q}$.

The relation~$\arrow{}$ for non-deterministic processes is
  finitely branching, but the relation~$\arrow{}$ for probabilistic
  processes is not. Following~\cite{SL94:concur,Seg95:thesis}, the
transition relation~$\rightarrow$ on distributions as given by
Definition~\ref{def-pr-transition-relation} allows for a probabilistic
combination of non-deterministic alternatives resulting in a so-called
combined transition. For example, for the process
$E = a \pref \mkern1mu (P \probc{\mytinyfrac12} Q) + a \pref \mkern1mu
(P \probc{\mytinyfrac13} Q)$ of~\cite{BS01:icalp}, we have that the
Dirac process
$\partial(E) = \partial( a \pref \mkern1mu (P \probc{\mytinyfrac12} Q)
+ a \pref \mkern1mu (P \probc{\mytinyfrac13} Q) )$ provides an
$a$-transition to $\den{P \probc{\mytinyfrac12} Q}$ as well as an
$a$-transition to $\den{P \probc{\mytinyfrac13} Q}$.
So, since we can represent the distribution~$\delta(E)$ by
$\delta(E) = \frac12 \delta(E) \oplus \frac12 \delta(E)$, the
distribution~$\delta(E)$ also has a combined transition
\begin{displaymath}
  \delta(E) 
  = \half \delta(E) \oplus \half \delta(E)
  \arrow{a} \half \den{P \probc{\mytinyfrac12} Q} \oplus \half \den{P
  \probc{\mytinyfrac13} Q} = \den{P \probc{\mytinyfrac{5}{12}} Q} .
\end{displaymath}
As noted in~\cite{Sto02:phd}, the ability to combine transitions is
crucial for obtaining transitivity of probabilistic process
equivalences that take internal actions into account.
\vspace{-1ex}

\paragraph{Example}
Referring to the examples of processes above, we have, e.g,
  \begin{align*}
    \mathbf{H_1} \colon \quad
    & \delta( a \pref \mkern1mu (P \probc{\mytinyfrac14} (P
      \probc{\mytinyfrac13} Q)) ) 
      \arrow{a} \den{P \probc{\mytinyfrac14} (P
      \probc{\mytinyfrac13} Q)} = 
      \half \den{P} \oplus \half \den{Q}
      \smallskip \\ 
    \mathbf{H_2} \colon \quad
    & \delta( a \pref \mkern1mu (P \probc{\mytinyfrac12} (Q
      \probc{\mytinyfrac12} Q)) )
      \arrow{a} \den{P \probc{\mytinyfrac12} (Q
      \probc{\mytinyfrac12} Q)} = 
      \half \den{P} \oplus \half \den{Q}
      \smallskip \\
    \mathbf{G_2} \colon \quad
    & a \pref \bigl(
      \partial \bigl( \tau \pref ( P \probc{\frac12} Q ) \bigl)
      \probc{\frac13} ( P \probc{\frac12} Q ) \bigl)
      \arrow{a}
      \delta \bigl( \tau \pref ( P \probc{\frac12} Q ) \bigl)
      \probc{\frac13} ( P \probc{\frac12} Q )
      \mkern1mu .
  \end{align*}
  Because a transition of a probabilistic process yields a
  distribution, the $a$-transitions of $\mathbf{H_1}$
  and~$\mathbf{H_2}$ have the same target. It is noted that
  $\mathbf{G_2}$ doesn't provide a further transition unless both its
  components $P$ and~$Q$ do so to match the transition of
  $\tau \pref ( P \probc{\frac12} Q )$.  \vspace{2ex}

\noindent
In preparation to the definition of the notion of branching
probabilistic {\bisimilarity} in Section~\ref{sec-props} we introduce
some notation.

\begin{definition}
  \label{def-weak-arrows}
  For $\mu, \mu' \mathbin\in \DistrE$ and $\alpha \mathbin\in \calA$
  we write $\mu \alphahidearrow \mu'$ iff (i)~$\mu \alphaarrow \mu'$,
  or (ii)~$\alpha = \tau$ and $\mu'=\mu$, or (iii)~$\alpha = \tau$ and
  there exist $\mu_1, \mu_2, \mu'_1, \mu'_2 \in \DistrE$ such that
  $\mu = \mu_1 \oplusr \mu_2$, $\mu' = \mu'_1 \oplusr \mu'_2$,
  \vspace{1pt} $\mu_1 \arrow{\tau} \mu_1'$ and $\mu_2 = \mu'_2$ for
  some $r \in (0,1)$.
\end{definition}

\noindent
Cases (i) and (ii) in the definition above correspond with the limits
$r = 1$ and $r = 0$ of case (iii).  We use $\Arrow{\,}$ to denote the
reflexive transitive closure of $\arrow{(\tau)}$.
A transition $\mu \tauhidearrow \mu'$ is called a partial transition,
and a transition $\mu \Rightarrow \mu'$ is called a weak transition.

\paragraph{Example}
  \begin{itemize}
  \item [(a)] According to Definition~\ref{def-weak-arrows} we have
    \vspace{-1ex}
    \begin{align*}
      &
      \textstyle\frac13 \delta( \tau \pref \mkern1mu (P
        \probc{\mytinyfrac12} Q)) \oplus 
      \textstyle\frac23 \den{ P \probc{\mytinyfrac12} Q }
      \  \arrow{(\tau)} \
      \textstyle\frac13 \den{P \probc{\mytinyfrac12} Q}
      \oplus
      \textstyle\frac23 \den{P \probc{\mytinyfrac12} Q}
      =
      \den{P \probc{\mytinyfrac12} Q}
      \mkern1mu .
    \end{align*}
    \vspace{-5ex}
  \item [(b)] There are typically multiple ways to construct a weak
    transition~$\Rightarrow$. Consider the weak transition
    $\frac12 \delta( \tau \pref \partial ( \tau \pref P )) \oplus
    \frac13 \delta( \tau \pref P ) \oplus \frac16 \den{P} \Arrow{\:}
    \den{P}$ which can be obtained, among uncountably many other
    possibilities, via
    \begin{align*}
      & \textstyle\frac12 \delta( \tau \pref \partial ( \tau \pref P ))
        \oplus
        \textstyle\frac13 \delta( \tau \pref P )
        \oplus
        \textstyle\frac16 \den{P}
        \tauhidearrow
        {} \smallskip \\
      & \qquad \qquad
        \textstyle\frac12 \delta ( \tau \pref P ))
        \oplus
        \textstyle\frac13 \delta( \tau \pref P )
        \oplus
        \textstyle\frac16 \den{P}
        =
        \textstyle\frac56 \delta( \tau \pref P )
        \oplus
        \textstyle\frac16 \den{P}
        \tauhidearrow
        \den{P},
      \intertext{or via}
      & \textstyle\frac12 \delta( \tau \pref \partial ( \tau \pref P ))
        \oplus
        \textstyle\frac13 \delta( \tau \pref P )
        \oplus
        \textstyle\frac16 \den{P}
        \tauhidearrow
        \textstyle\frac12 \delta( \tau \pref \partial ( \tau \pref P ))
        \oplus
        \textstyle\frac13 \delta( P )
        \oplus
        \textstyle\frac16 \den{P}
        = {} \smallskip \\
      & \qquad \qquad
        \textstyle\frac12 \delta( \tau \pref \partial ( \tau \pref P ))
        \oplus
        \textstyle\frac12 \delta( P )
        \tauhidearrow
        \textstyle\frac12 \delta ( \tau \pref P )
        \oplus
        \textstyle\frac12 \den{P}
        \tauhidearrow
        \textstyle\frac12 \den{P}
        \oplus
        \textstyle\frac12 \den{P}
        =
        \den{P}.
    \end{align*}
  \item [(c)] The distribution
    $\textstyle\frac12 \delta( \tau \pref \partial(\astop + \bstop)) \oplus
    \textstyle\frac12 \delta(a \pref \partial( \cstop ))$ doesn't admit a
    $\tau$-transition nor an $a$-transition. However, we have
      \begin{align*}
        & \textstyle\frac12 \delta( \tau \pref \partial(\astop + \bstop))
          \oplus
          \textstyle\frac12 \delta(a \pref \partial( \cstop ))
          \tauhidearrow
          {} \smallskip \\
        & \qquad \qquad
          \textstyle\frac12 \partial(\astop + \bstop)
          \oplus
          \textstyle\frac12 \delta(a \pref \partial( \cstop ))
          \arrow{\mkern2mu a}
          \textstyle\frac12 \delta(\bfzero) \oplus \textstyle\frac12
          \delta( \cstop ) 
          \mkern1mu .
      \end{align*}
    \end{itemize}

\noindent
The following lemma states that the transitions $\arrow{\alpha}$,
$\arrow{(\alpha)}$, and~$\Rightarrow$ of
Definitions~\ref{def-pr-transition-relation} and~\ref{def-weak-arrows}
can be probabilistically composed.

\begin{lemma}
  \label{lem-composition}
  Let, for a finite index set $I$, $\mu_i,\mu'_i \in \DistrE$ and
  $p_i \geqslant 0$ such that $\sum_{i\in I} \: p_i = 1$.
  \begin{enumerate} [(a)]
  \item [(a)] If $\mu_i \arrow{\alpha} \mu'_i$ for all $i\in I$, then
    $\bigoplusiinI \: p_i \stardot \mu_i \arrow{\alpha} \bigoplusiinI
    \: p_i \stardot \mu'_i$.
  \item [(b)] If $\mu_i \arrow{(\tau)} \mu'_i$ for all $i\in I$, then
    $\bigoplusiinI \: p_i \stardot \mu_i \arrow{(\tau)} \bigoplusiinI
    \: p_i \stardot \mu'_i$.
  \item [(c)] If $\mu_i \Arrow{} \mu'_i$ for all $i\in I$, then
    $\bigoplusiinI \: p_i \stardot \mu_i \Arrow{\;} \bigoplusiinI \:
    p_i \stardot \mu'_i$.
  \end{enumerate}
\end{lemma}

\begin{proof}
  Let $\mu = \bigoplusiinI \: p_i \stardot \mu_i$ and
  $\mu' = \bigoplusiinI \: p_i \stardot \mu'_i$. Without loss of
  generality, we may assume that $p_i>0$ for all~$i \in I$.

  (a) 
  Suppose $\mu_i \arrow{\alpha} \mu'_i$ for~all $i \in I$. Then, by
  Definition~\ref{def-pr-transition-relation},
  $\mu_i = \bigoplusjinJi \: p_{ij} \stardot E_{ij}$,
  $\mu'_i = \bigoplusjinJi \: p_{ij} \stardot \eta_{ij} $, and
  $E_{ij} \alphaarrow \eta_{ij}$ for~$j \in J_i$ for a suitable index
  set~$J_i$, $p_{ij} > 0$ and $\eta_{ij} \in \DistrE$.
  Define the index set~$K$ and probabilities~$q_k$ for~$k \in K$ by
  $K = \lc (i,j) \mid i \in I ,\, j \in J_i \rc$ and
  $q_{(i,j)} = p_i \mkern1mu p_{ij}$ for~$(i,j) \in K$, so that
  $\sum_{k {\in} K} \: q_k = 1$.
  Then we have $\mu = \bigopluskinK \: q_k \stardot E_{ij}$ and
  $\mu' = \bigopluskinK \: q_k \stardot \eta_{ij}$. Therefore, by
  Definition~\ref{def-pr-transition-relation}, it follows that
  $\mu \arrow{\alpha} \mu'$.

  (b)~Let $\mu_i \arrow{(\tau)} \mu'_i$ for all $i\in I$.  Then, for
  all $i \in I$, by Definition~\ref{def-weak-arrows}, there exists
  $r_i \in [0,1]$ and
  $\mu_i^{\rm stay}, \mu_i^{\rm go}, \mu''_i
  \mathbin\in \DistrE$, such that
  $\mu_i \mathbin= \mu_i^{\rm stay} \! \probc{r_i} \mu_i^{\rm go}$,
  $\mu'_i \mathbin= \mu_i^{\rm stay} \! \probc{r_i} \mu''_i$, and
  either $r_i \mathbin{=} 1$ or $\mu_i^{\rm go} \arrow{\tau} \mu_i''$.
  In case~$r_i = 0$ for all $i \in I$, we have that
  $\mu_i \arrow{\tau} \mu'_i$ for all $i \in I$, and thus
  $\mu \arrow{\tau} \mu'$ by the first claim of the lemma, and
  $\mu \tauhidearrow \mu'$ by Definition~\ref{def-weak-arrows}(i).
  In case $r_i = 1$ for all $i \in I$, we have $\mu' = \mu$ and thus
  $\mu \tauhidearrow \mu'$ by Definition~\ref{def-weak-arrows}(ii).
  Otherwise, let $I' := \lc i \in I \mid r_i<1 \rc$,
  $r = \sumiinI \: p_i\mycdot r_i$,
  $\mu^{\rm stay} := \bigoplusiinI \: \frac{p_i\mycdot r_i}{r}
  \stardot \mu_i^{\rm stay}$,
  $\mu^{\rm go} := \bigoplus_{i\in I'} \: \frac{p_i\mycdot
    (1-r_i)}{1-r} \stardot \mu_i^{\rm go}$ and
  $\mu'' := \bigoplus_{i\in I'} \: \frac{p_i\mycdot (1-r_i)}{1-r}
  \stardot \mu''_i$. Then $\mu^{\rm go} \arrow\tau \mu''$ by the first
  claim of the lemma. Moreover,
  $\mu = \mu^{\rm stay} \! \probc{r} \mu^{\rm go}$,
  $\mu' = \mu^{\rm stay} \! \probc{r} \mu''$ and $r \in (0,1)$.  So
  $\mu \tauhidearrow \mu'$ by Definition~\ref{def-weak-arrows}(iii).

  (c)~Let $\mu_i \Arrow{} \mu'_i$ for all $i\in I$.  As $I$ is
  finite and $\Arrow{}$ is reflexive, there exists an $n \in \IN$
  such that
  $\mu_i = \mu_i^{(0)} \arrow{(\tau)} \mu_i^{(1)} \arrow{(\tau)} \dots
  \arrow{(\tau)} \mu_i^{(n)} = \mu'_i$ for all $i\in I$. Now
  $\mu \Arrow{} \mu'$ follows by $n$ applications of the second
  statement of the lemma.
\end{proof}

\noindent
Likewise, the next lemma allows \emph{probabilistic decomposition} of
 transitions $\arrow{\alpha}$, $\arrow{(\alpha)}$ and $\Arrow$.

 \begin{lemma}
   \label{lem-decomposition}
   Let $\mu,\mu'\in \DistrE$ and
   $\mu = \bigoplusiinI \: p_i \stardot \mu_i$ with $p_i>0$
   for~$i \in I$.
   \begin{itemize}
   \item [(a)] If $\mu \arrow{\alpha} \mu'$, then there are $\mu'_i$
     for $i \in I$ such that $\mu_i \arrow{\alpha} \mu'_i$ for 
     $i \in I$ and $\mu' = \bigoplusiinI \: p_i \stardot \mu'_i$.
   \item [(b)] If $\mu \arrow{(\tau)} \mu'$, then there are $\mu'_i$ for
     $i \in I$ such that $\mu_i \arrow{(\tau)} \mu'_i$ for 
     $i \in I$ and $\mu' = \bigoplusiinI \: p_i \stardot \mu'_i$.
   \item [(c)] If $\mu \Arrow{} \mu'$, then there are $\mu'_i$ for
     $i \in I$ such that $\mu_i \Rightarrow \mu'_i$ for $i \in I$
     and $\mu' = \bigoplusiinI \: p_i \stardot \mu'_i$.
   \end{itemize}
 \end{lemma}

 \begin{proof}
   (a) Suppose $\mu \arrow{\alpha} \mu'$.  By
   Definition~\ref{def-pr-transition-relation}
   $\mu = \bigoplusjinJ \: q_j \stardot E_j$,
   $\mu' = \bigoplusjinJ \: q_j \stardot \eta_j$, and
   $E_j \alphaarrow \eta_j$ for all $j \in J$, for suitable index
   set~$J$, $q_j > 0$, $E_j \in \calE$, and~$\eta_j \in \DistrE$. By
   Lemma~\ref{flexibelDelen} there are $r_{ij} \geqslant 0$ and
   $\varrho_{ij} \in \DistrE$ such that $\sumjinJ \: r_{ij} = p_i$ and
   $p_i \mkern1mu \mu_i = \bigoplusjinJ \: {r_{ij}} \varrho_{ij}$
   for~$i \in I$, and $\sumiinI \: r_{ij} = q_j$ and
   $q_j \stardot \delta(E_j) = \bigoplusiinI \: r_{ij}
   \varrho_{ij}$ for all~$j \in J$. Hence, $\varrho_{ij}=\delta(E_j)$
   for $i \in I$,~$j \in J$.

   For all $i \in I$, let
   $\mu'_i = \bigoplusjinJ \: ({r_{ij}}/{p_i}) \mkern1mu \eta_j$. Then
   $\mu_i \arrow\alpha \mu'_i$, for all $i \in I$, by
   Lemma~\ref{lem-composition}(a). Moreover, it holds that
   $\bigoplusiinI \: p_i \mkern1mu \mu'_i = \bigoplusiinI \: p_i
   \stardot \bigoplusjinJ \: ({r_{ij}}/{p_i}) \mkern1mu \eta_j =
   \bigoplusjinJ \: \bigoplusiinI \: r_{ij} \stardot \eta_j =
   \bigoplusjinJ \: q_{j} \stardot \eta_j = \mu'$.

  (b)~Suppose $\mu \arrow{(\tau)} \mu'$. By
  Definition~\ref{def-weak-arrows}, either (i)~$\mu \tauarrow \mu'$,
  or (ii)~$\mu'=\mu$, or (iii)~there exist
  $\nu_1, \nu_2, \nu'_1, \nu'_2 \in \DistrE$ such that
  $\mu = \nu_1 \oplusr \nu_2$, $\mu' = \nu'_1 \oplusr \nu'_2$,
  $\nu_1 \arrow{\tau} \nu_1'$ and $\nu_2 = \nu'_2$ for some
  $r \in (0,1)$.  In case~(i), the required $\mu'_i$ exist by the
  first statement of this lemma.  In case~(ii) one can simply take
  $\mu'_i := \mu_i$ for all $i\in I$.  Hence assume that case~(iii)
  applies.  Let $J := \{1,2\}$, $q_1:=r$ and $q_2:=1-r$.  By
  Lemma~\ref{flexibelDelen} there are $r_{ij} \in [0,1]$ and
  $\varrho_{ij} \in \DistrE$ with $\sum_{j\in J} \: r_{ij} = p_i$ and
  $\mu_i = \bigoplusjinJ \: \frac{r_{ij}}{p_i} \stardot \varrho_{ij}$
  for all~$i \in I$, and $\sum_{i\in I} \: r_{ij}=q_j$ and
  $\nu_j = \bigoplusiinI \: \frac{r_{ij}}{q_j} \stardot \varrho_{ij}$
  for all~$j \in J$.

  Let $I' := \{i\in I\mid r_{i1}>0\}$.\vspace{1pt} Since
  $\nu_1 = \bigoplus_{i\in I'} \: \frac{r_{i1}}{r} \stardot
  \varrho_{i1}\arrow{\tau} \nu_1'$, by the first statement of the
  lemma, for all $i \in I'$ there are $\varrho_{i1}'$ such that
  $\varrho_{i1} \arrow{\tau} \varrho_{i1}'$ and
  $\nu'_1=\bigoplus_{i\in I'} \: \frac{r_{i1}}{r} \stardot
  \varrho_{i1}'$.\vspace{2pt} For all $i \in I{\setminus}I'$ pick
  $\varrho'_{i1} \in \DistrE$ arbitrarily.  It follows that
  \(\mu_i = \varrho_{i1} \probc{\frac{r_{i1}}{p_i}} \varrho_{i2}
  \tauhidearrow \varrho'_{i1} \probc{\frac{r_{i1}}{p_i}} \varrho_{i2}
  =: \mu'_i\) for all $i \in I$.  Moreover,
  \centerline{\begin{math}
    \bigoplusiinI \: p_i \stardot \mu'_i =
    \bigoplusiinI \: p_i \stardot (\varrho'_{i1}
    \probc{\frac{r_{i1}}{p_i}} \varrho_{i2}) = 
    (\bigoplusiinI \: \frac{r_{i1}}{r} \stardot \varrho'_{i1}) \oplusr
    (\bigoplusiinI \: \frac{r_{i2}}{1-r} \stardot \varrho_{i2}) =
    \nu'_1 \oplusr \nu_2 = \mu'
    \; .
  \end{math}}
  (c)~The last statement follows by transitivity from the second one.
 \end{proof}

%% file: properties.tex

\section{Branching probabilistic bisimilarity}
\label{sec-props}

In this section we recall the notion of branching probabilistic
bisimilarity~\cite{GGV19:fc}. The notion is based on a
decomposability property due to~\cite{DGHM09} and a transfer property.

\begin{definition}[Branching probabilistic {\bisimilarity}]
  \label{def-probabilistic-branching-bisimilar}
  \mbox{}
  \begin{enumerate} [(a)]
  \item A relation $\calR \subseteq \DistrE \times \DistrE$ is called
    \textit{weakly decomposable} iff it is symmetric and for all
    $\mu, \nu \in \DistrE$ such that $\mu \mopcalR \nu$ and
    $\mu = \bigoplusiinI \: p_i \stardot \mu_i$ there are
    $\nubar, \nu_i \in \DistrE$, for~$i \in I$, such that \smallskip
    \\ \centerline{\begin{math} \nu \Arrow{} \nubar,\ \mu \mopcalR
        \nubar,\ \nubar = \bigoplusiinI \: p_i \stardot \nu_i,\
        \text{and\,} \ \mu_i \mopcalR \nu_i \,\ \text{for
          all~$i \in I$.}
    \end{math}}
\item A relation $\calR \subseteq \DistrE \times \DistrE$ is called a
  \emph{branching} probabilistic bisimulation relation
  iff it is weakly decomposable and for all
  $\mu, \nu \in \DistrE$ with $\mu \mopcalR \nu$ and
  $\mu \arrow{\alpha} \mu'$, there are $\nubar, \nu' \in \DistrE$ such
  that\vspace{-1ex} \smallskip \\ \centerline{\begin{math} \nu
      \Arrow{} \nubar, \ \nubar \alphahidearrow \nu', \ \mu \mopcalR
      \nubar, \ \text{and} \,\ \mu' \mopcalR \nu'.
    \end{math}}
\item Branching probabilistic {\bisimilarity}
  ${\brbisim} \subseteq \DistrE \times \DistrE$ is defined as the
  largest branching probabilistic bisimulation relation on~$\DistrE$.
  \end{enumerate}
\end{definition}

\noindent
Note that branching probabilistic {\bisimilarity} is well-defined
following the usual argument that any union of branching probabilistic
bisimulation relations is again a branching probabilistic bisimulation
relation. In particular, (weak) decomposability is preserved under
arbitrary unions.
As observed in~\cite{GGV19:tue}, branching probabilistic bisimilarity
is an equivalence relation.

Two non-deterministic processes are considered to be branching
probabilistic bisimilar iff their Dirac distributions are, i.e., for
$E, F \in \calE$ we have $E \brbisim F$ iff
$\delta(E) \brbisim \delta(F)$.
Two probabilistic processes are considered to be branching
probabilistic bisimilar iff their associated distributions
over~$\calE$ are, i.e., for $P,Q \in \calP$ we have $P \brbisim Q$ iff
$\den P \brbisim \den Q$.

For a set $M\subseteq \DistrE$, the convex closure $cc(M)$ is defined
by
\begin{displaymath}
  cc(M) = \{\bigoplusiinI p_i\mu_i\mid \sum_{i\in I}p_i=1, ~ \mu_i\in M, ~
  I \textrm{ a finite index set}\}.
\end{displaymath}
For a relation $\calR \subseteq \DistrE \times \DistrE$ the convex
closure of $\calR$ is defined
by
\begin{displaymath}
  \mathit{cc}(\calR) =
  \lc \langle\bigoplusiinI p_i\mu_i,\bigoplusiinI p_i\nu_i\rangle
  \mid \mu_i\calR \nu_i, ~\sum_{i\in I}p_i=1, ~
  I\textrm{ a finite index set} \rc \mkern1mu .
\end{displaymath}

\noindent
The notion of weak decomposability has been adopted
from~\cite{Hen12:facj,LV16}. The underlying idea stems
from~\cite{DGHM09}. Weak decomposability provides a convenient
dexterity to deal with combined transitions as well as with
sub-distributions. For example, regarding sub-distributions, to
distinguish the probabilistic process
$\frac12 \partial( a \pref \partial(\bfzero)) \oplus \frac12 \partial(
b \pref \partial(\bfzero))$ from $\partial(\bfzero)$ a branching
probabilistic bisimulation relation relating
$\frac12 \delta( a \pref \partial(\bfzero)) \oplus \frac12 \delta( b
\pref \partial(\bfzero))$ and~$\delta(\bfzero)$ is by weak
decomposability also required to relate
$\delta( a \pref \partial(\bfzero))$ and $\delta( b \pref \partial(\bfzero))$
to subdistributions of a weak
descendant of~$\delta(\bfzero)$, which can only be
$\delta(\bfzero)$~itself. Since $\delta( a \pref \partial(\bfzero))$
has an $a$-transition while $\delta(\bfzero)$ has not, and
similar for a $b$-transition of
$\delta( b \pref \partial(\bfzero))$, it follows that
$\frac12 \partial( a \pref \partial(\bfzero)) \oplus \frac12 \partial(
b \pref \partial(\bfzero))$ and $\partial(\bfzero)$ are not branching
probabilistic bisimilar.

By comparison, on finite processes, as used in this paper, the notion
of branching probabilistic bisimilarity of Segala \&
Lynch~\cite{SL94:concur} can be defined in our framework exactly as in
(b) and~(c) above, but taking a decomposable instead of a weakly
decomposable relation, i.e.\ if $\mu \mopcalR \nu$ and
$\mu = \bigoplusiinI \: p_i \mu_i$ then there are~$\nu_i$ for~$i \in I$
such that $\nu = \bigoplusiinI \: p_i \nu_i$ and
$\mu_i \mopcalR \nu_i$ for~$i \in I$. This yields a strictly finer
equivalence.

\noindent
\paragraph{Example}
  \label{ex-prob-bisim}
  \begin{enumerate} [(a)]
  \item The distributions
    $\delta( \mathbf{G_1} ) = \delta( a \pref \mkern1mu ( P
    \probc{\mytinyfrac12} Q ))$
    and~$\delta( \mathbf{G_2} ) = \delta( a \pref \mkern1mu ( \partial
    ( \tau \pref \mkern1mu ( P \probc{\mytinyfrac12} Q ) )
    \probc{\mytinyfrac13} ( P \probc{\mytinyfrac12} Q ) ))$ both admit
    at the top level
    an $a$-transition only:
    \begin{align*}
      \delta( a \pref \mkern1mu ( P \probc{\mytinyfrac12} Q ))
      & \arrow{a}
        \half \den{P} \oplus \half \den{Q} \\
      \delta( a \pref \mkern1mu ( \partial (
      \tau \pref \mkern1mu ( P \probc{\mytinyfrac12} Q ) )
      \probc{\mytinyfrac13} ( P \probc{\mytinyfrac12} Q ) ))
      & \arrow{a}
        \ttfrac13 \delta ( \tau \pref \mkern1mu
        ( P \probc{\mytinyfrac12} Q ) )
        \oplus
        \ttfrac13 \den{P} \oplus
        \ttfrac13 \den{Q}
        \mkern1mu .
    \end{align*}
    Let the relation $\calR$ contain the pairs
    \begin{displaymath}
      \langle
      \delta( \tau \pref \mkern1mu ( P \probc{\mytinyfrac12} Q ) ),
      \ttfrac12 \den{P} \oplus \ttfrac12 \den{Q}
      \rangle
      \quad \text{and} \quad
      \langle \mu, \mu \rangle
      \text{ for $\mu \in \DistrE$}.
    \end{displaymath}
    The symmetric closure $\calR^\dagger$ of $\calR$
    is clearly a branching probabilistic bisimulation relation. We
    claim that therefore also its convex closure
    $\mathit{cc}(\calR^\dagger)$ is a branching probabilistic
    bisimulation relation.
    Considering that
    $\langle \delta( \tau \pref \mkern1mu ( P \probc{\mytinyfrac12} Q )
    ), \frac12 \den{P} \oplus \frac12 \den{Q} \rangle$ and
    $\langle \frac12 \den{P} \oplus \frac12 \den{Q}, \frac12
    \den{P} \oplus \frac12 \den{Q} \rangle$ are in $\calR$, we have that
  \begin{displaymath}
    \langle
    \ttfrac13 \delta( \tau \pref \mkern1mu ( P \probc{\mytinyfrac12} Q
    ) \oplus
    \ttfrac23 ( \ttfrac12 \den{P} \oplus \ttfrac12 \den{Q} ),
    \ttfrac13( \ttfrac12 \den{P} \oplus \ttfrac12 \den{Q} ) \oplus
    \ttfrac23 ( \ttfrac12 \den{P} \oplus \ttfrac12 \den{Q} ) )    
    \rangle \in \mathit{cc}(\calR^\dagger) {\mkern1mu .}
  \end{displaymath}
  Adding the pair of processes
  $\langle \delta( a \pref \mkern1mu ( P \probc{\mytinyfrac12} Q )),
  \delta( a \pref \mkern1mu ( \partial ( \tau \pref \mkern1mu ( P
  \probc{\mytinyfrac12} Q ) ) \probc{\mytinyfrac13} ( P
  \probc{\mytinyfrac12} Q ) )) \rangle$ and closing for symmetry,
  then yields a branching probabilistic bisimulation relation relating
  $\delta( \mathbf{G_1} )$ and $\delta( \mathbf{G_2} )$.
    
\item [(b)] The $a$-derivatives of $\mathbf{I_1}$ and~$\mathbf{I_2}$,
  i.e.\ the distributions $I'_1 = \delta( b \pref P + \tau \pref Q )$
  and
  $I_2' = \delta( \tau \pref \partial( b \pref P + \tau \pref Q ) + b
  \pref P + \tau \pref Q)$ are branching probabilistic bisimilar. A
  $\tau$-transition of~$I'_2$ partially based on its left branch, can
  be simulated by~$I'_1$ by a partial transition:
  \begin{displaymath}
    \begin{array}{r@{\:}c@{\:}l@{\quad}c@{\quad}l}
      I'_2
      & \myeq & r {\cdot} \den{ I'_2 } \oplus (1-r) {\cdot} \den{I'_2}
      & \arrow{~\tau~}
      & r {\cdot} \mkern1mu
        \delta( b \pref P + \tau \pref Q ) \oplus (1-r) {\cdot}
        \den{ Q } \smallskip \\
      I'_1
      & \myeq & r {\cdot} \den{ I'_1 } \oplus (1-r) {\cdot} \den{I'_1}
      & \arrow{(\tau)}
      & r {\cdot} \den{ I'_1 } \oplus (1-r) {\cdot} \den{ Q }
        \; \myeq \;
        r {\cdot}
        \delta( b \pref P + \tau \pref Q ) \oplus (1-r) {\cdot}
        \den{ Q }
        \mkern1mu .
    \end{array}
  \end{displaymath}
  A $\tau$-transition of~$I'_1$ can be directly simulated by~$I'_2$ of
  course. It follows that the relation
  $\calR = \singleton{ \langle{ \delta( \mathbf{I_1} ), \delta(
      \mathbf{I_2} ) }\rangle, \langle{ I'_1, I'_2 }\rangle }^\dagger
  \cup \lc \langle \mu, \mu \rangle \mid \mu \in \DistrE \rc$, the
  symmetric relation containing the pairs mentioned and the diagonal
  of~$\DistrE$, constitutes a branching probabilistic bisimulation
  relation containing $\mathbf{I_1}$ and~$\mathbf{I_2}$.
  \end{enumerate}

\noindent
In the sequel we frequently need that probabilistic composition
respects branching probabilistic bisimilarity of distributions, i.e.\
if, with respect to some index set~$I$, we have distributions $\mu_i$
and~$\nu_i$ such that $\mu_i \brbisim \nu_i$ for~$i \in I$, then
also~$\mu \brbisim \nu$ for the distributions
$\mu = \bigoplusiinI \: p_i \mu_i$ and
$\nu = \bigoplusiinI \: p_i \nu_i$. The property directly follows from
the following lemma, which is proven in~\cite{GGV19:tue}.

\begin{lemma}
  \label{lemma-congruence-for-prob-composition}
  Let distributions $\mu_1,\mu_2,\nu_1,\nu_2\in \DistrE$ and
  $0 \leqslant r \leqslant 1$ be such that $\mu_1 \brbisim \nu_1$ and
  $\mu_2 \brbisim \nu_2$. Then it holds that
  $\mu_1 \probc{r} \mu_2 \brbisim \nu_1 \probc{r} \nu_2$.
\end{lemma}

\noindent
We apply the above property in the proof of the next result. In
the sequel any application of
Lemma~\ref{lemma-congruence-for-prob-composition} will be done
tacitly.

\begin{lemma}
  \label{lemma-weak-transfer}
  Let $\mu, \nu \in \DistrE$ such that $\mu \brbisim \nu$ and
  $\mu \Arrow{\:} \mu'$ for some $\mu' \in \DistrE$. Then there are
  $\nu' \in \DistrE$ such that $\nu \Arrow{\:} \nu'$ and
  $\mu' \brbisim \nu'$.
\end{lemma}

\begin{proof}
  We check that a partial transition $\mu \arrow{(\tau)} \mu'$ can be
  matched by~$\nu$ given $\mu \brbisim \nu$. So, suppose
  $\mu = \mu_1 \probc{r} \mu_2$, $\mu_1 \arrow{\tau} \mu'_1$, and
  $\mu' = \mu'_1 \probc{r} \mu_2$. By weak decomposability
  of~$\brbisim$ we can find distributions~$\nubar, \nu_1, \nu_2$ such
  that $\nu \Arrow{\:} \nubar = \nu_1 \probc{r} \nu_2$ and
  $\mu \brbisim \nubar$, $\nu_1 \brbisim
  \mu_1$,~$\nu_2 \brbisim \mu_2$. Choose
  distributions~$\nubar_1, \nubar'_1$ such that
  $\nu_1 \Arrow{\:} \nubar_1 \arrow{(\tau)} \nu'_1$ and
  $\nubar_1 \brbisim \mu_1$, $\nu'_1 \brbisim \mu'_1$. Put
  $\nu' = \nu'_1 \probc{r} \nu_2$. Then $\nu \Arrow{\:} \nu'$,
  using Lemma~\ref{lem-composition}c, and we
  have by Lemma~\ref{lemma-congruence-for-prob-composition} that
  $\nu' = \nu'_1 \probc{r} \nu_2 \brbisim \mu'_1 \probc{r} \mu_2 =
  \mu'$ since $\nu'_1 \brbisim \mu'_1$ and $\nu_2 \brbisim \mu_2$.
\end{proof}

%% file: continuity.tex

\section{Branching probabilistic bisimilarity is continuous}
\label{sec-continuity}

Fix a finite set of non-deterministic processes
$\calF\mathbin\subseteq \calE$ that is \emph{transition closed}, in
the sense that if $E \in \calF$ and
$E \arrow\alpha \bigoplusiinI p_i{\cdot}F_i$ then also
$F_i \in \calF$.
Consequently, if $\mu \in \DistrF$ and $\mu \alphahidearrow \mu'$
then~$\mu' \in \DistrF$. Also, if $\mu \in \DistrF$ and
$\mu \Arrow{} \mubar$ then~$\mubar \in \DistrF$.
By Theorem~\ref{thm-seq-compact}
$\DistrF$ is a sequentially compact subspace of the complete metric
space~$\DistrE$, meaning that every sequence
$( \mu_i )_{i{=}0}^\infty$ in~$\DistrF$ has a subsequence
$( \mu_{i_k} )_{k{=}0}^\infty$ such that
$\lim_{\, k \rightarrow \infty} \: \mu_{i_k} = \mu$ for some
distribution~$\mu \in \DistrF$.
In particular, if $\lim_{\, i \rightarrow \infty} \: \mu_i = \mu$ and
$\mu_i \in \DistrF$, then also~$\mu \in \DistrF$, i.e.\ $\DistrF$ is a
closed subset of~$\DistrE$.
Due to the finitary nature of our process algebra, each distribution
$\mu \in \Distr(\calE)$ occurs in $\Distr(\calF)$ for some such
$\calF$, based on~$\spt(\mu)$. 

In the following three lemmas we establish a number of continuity
results. Assume
$\lim_{\, i \rightarrow \infty} \nu_i = \nu$. Then Lemma~\ref{partial
  transitions closed 1} states that, for a Dirac
distribution~$\delta(E)$, if $\delta(E) \arrow\alpha \nu_ i$
for~$i \in \IN$ then also $\delta(E) \arrow\alpha
\nu$. Lemma~\ref{partial transitions closed 2} extends this and shows
that, for a general distribution~$\mu$, if $\mu \arrow\alpha \nu_ i$
for~$i \in \IN$ then $\mu \arrow\alpha \nu$. Finally,
Lemma~\ref{partial transitions closed 3} establishes the limit case:
if $\lim_{\, i \rightarrow \infty} \mu_i = \mu$ and
$\mu _i\arrow\alpha \nu_ i$ for~$i \in \IN$ then
$\mu \arrow\alpha \nu$.

\begin{lemma}
  \label{partial transitions closed 1}
  Let $E \in \calF$ be a non-deterministic process, $\alpha \in \calA$ an
  action, $(\nu_i)_{i=0}^\infty \in \DistrF^\infty$ an infinite
  sequence in~$\DistrF$, and $\nu \in \DistrF$ a distribution
  \vspace{2pt} satisfying $\lim_{\, i \rightarrow \infty} \nu_i =
  \nu$. If, for all $i \in \IN$, $\delta(E) \alphahidearrow \nu_i$
  then it holds that $\delta(E) \alphahidearrow \nu$.
\end{lemma}

\begin{proof}
  For $E \in \calF$ and $\alpha \in \calA$, define
  $E \after \alpha = \mathit{cc}(\lc \mu \mid  E \arrow\alpha \mu\})$,
  pronounced $E$ `after'~$\alpha$, to be the
  convex closure in~$\DistrE$ of all distributions that can be
  reached from~$E$ by an $\alpha$-transition.  Then
  $\delta(E) \arrow\alpha \nu$ iff $\nu \in E \after \alpha$.
  Recall that transitions for non-deterministic processes are not
  probabilistically combined. See
  Definition~\ref{def-pr-transition-relation}.
  Since $E \after \alpha \subseteq \DistrF$ is the convex closure of a
  finite set of distributions, it is certainly closed in the
  space~$\DistrF$.  Since it holds that $\delta(E) \alphaarrow \nu_i$
  for all $i \in \IN$, one has $\nu_i \in E \after \alpha$
  for~$i \in \IN$. Hence, $\lim_{\, i \rightarrow\infty} \nu_i = \nu$
  implies that $\nu \in E \after \alpha$, i.e.\
  $\delta(E) \alphaarrow \nu$.

  For $E \in \calF$, define
  $E \after (\tau) := \mathit{cc}(\{\mu \mid E \arrow\tau \mu\}\cup\{E\})$.
  Then $\delta(E) \tauhidearrow\nu$ iff
  $\nu \in E \after (\tau)$. The set
  $E \after (\tau) \subseteq \DistrF$ is closed, and thus
  $\nu_i \in E \after (\tau)$ implies $\nu \in E \after (\tau)$, which
  means $\delta(E) \tauhidearrow \nu$.
\end{proof}

\noindent
The above result for Dirac distributions holds for
general distributions as well.

\begin{lemma}
  \label{partial transitions closed 2}
  Let $\mu, \nu \in \Distr(\calF)$, $\alpha\in\calA$,
  $(\nu_i)_{i{=}0}^\infty \in \Distr(\calF)^\infty$, and assume
  $\lim_{\, i \rightarrow\infty} \nu_i = \nu$. \vspace{2pt} If it
  holds that $\mu \alphahidearrow \nu_i \,$ for all $i \in \IN$, then
  also $\mu \alphahidearrow \nu$.
\end{lemma}

\begin{proof}
  Suppose $\mu \alphahidearrow \nu_i$ for all $i \in I$.  Let
  $\mu = \bigoplus_{j=1}^k p_j\mycdot E_j$. \vspace{-1pt} By
  Lemma~\ref{lem-decomposition}, for all $i \in \IN$ and
  $1 \leqslant j \leqslant k$ there are $\nu_{ij}$ such that
  $\delta(E_j) \alphahidearrow \nu_{ij}$ and
  $\nu_i = \bigoplus_{j=1}^k p_j \mycdot \nu_{ij}$.  The countable
  sequence $(\nu_{i 1},\nu_{i 2},\dots,\nu_{i k})_{i=0}^\infty$ of
  $k$-dimensional vectors of probability distributions need not have a
  limit.  However, by the sequential compactness of
    $\DistrF$ this sequence has an infinite subsequence in which the
    first components $\nu_{i_1}$ converge to a limit~$\eta_1$. That
    sequence in turn has an infinite subsequence in which also the
    second components $\nu_{i_2}$ converge to a limit~$\eta_2$.  Going
    on this way, one finds a subsequence
  $(\nu_{i_h 1},\nu_{i_h 2},\dots,\nu_{i_h k})_{h=0}^\infty$ of
  $(\nu_{i 1},\nu_{i 2},\dots,\nu_{i k})_{i=0}^\infty$ for
  $i_0< i_1 < \dots$ that has a limit, say
  \(\lim_{\, h\rightarrow\infty}(\nu_{i_h 1},\nu_{i_h
    2},\dots,\nu_{i_h k})= (\eta_1,\eta_2,\dots,\eta_k)\). Using that
  $\lim_{\, h \rightarrow\infty} \nu_{i_h}=\nu$, one obtains
  $\nu = \bigoplus_{j=1}^k p_j \mycdot \eta_j$.  For each
  $j=1,\dots,k$, by Lemma~\ref{partial transitions closed 1}, since
  $\delta(E_j) \alphahidearrow \nu_{ij}$ for all $i \in I$ and
  $\lim_{\, h \rightarrow\infty} \nu_{i_hj}=\eta_j$, we conclude that
  $\delta(E_j) \alphahidearrow \eta_j$.  Thus, by Lemma~\ref{lem-composition},
  $\mu = \bigoplus_{j=1}^k p_j\mycdot E_j \alphahidearrow
  \bigoplus_{j=1}^k p_j \mycdot \eta_j =\nu$.
\end{proof}

\noindent
Next, we consider a partial transition over a convergent sequence of
distributions.

\begin{lemma}
  \label{partial transitions closed 3}
  Let
  $(\mu_i)_{i{=}0}^\infty, (\nu_i)_{i{=}0}^\infty \in
  \Distr(\calF)^\infty$ such that
  $\lim_{\, i \rightarrow\infty} \mu_i = \mu$ and
  $\lim_{\, i \rightarrow\infty} \nu_i = \nu$. \vspace{2pt} If it
  holds that $\mu_i \alphahidearrow \nu_i$ for all $i \in \IN$, then
  also $\mu \alphahidearrow \nu$.
\end{lemma}

\begin{proof}
  Since $\lim_{\, i \rightarrow \infty} \mkern1mu \mu_i= \mu$, we can
  write
  $\mu_i= (1-r_i) \mu \mathrel\oplus r_i \mkern1mu \mu''_i$, for
  suitable $\mu''_i \in \DistrF$ and $r_i \geqslant 0$\linebreak[4] such that
  $\lim_{\, i \rightarrow \infty} \mkern1mu r_i = 0$, as guaranteed by
  Lemma~\ref{lem-representation-of-limit}.
  Because $\mu_i \alphahidearrow \nu_i$, by
  Lemma~\ref{lem-decomposition} there are
  distributions $\nu'_i, \nu''_i \in \DistrF$ for~$i \in \IN$ such
  that $\nu_i = (1-r_i) \nu'_i \mathrel\oplus r_i \nu''_i$,
  $\mu \alphahidearrow \nu'_i$, and $\mu''_i \alphahidearrow \nu''_i$.
  \vspace{2pt}
  We have $\lim_{\, i \rightarrow \infty} \nu'_i = \nu$ as well, since
  \vspace{2pt} $\lim_{\, i \rightarrow \infty} r_i = 0$.
  Thus, $\lim_{\, i \rightarrow \infty} \nu'_i = \nu$ and $\mu
  \alphahidearrow \nu'_i$ for~$i \in \IN$.
  Therefore, it follows by Lemma~\ref{partial transitions closed 2}
  that $\mu \alphahidearrow \nu$.
\end{proof}

\noindent
For $\mu, \nu \in \DistrF$, we write $\mu \Rightarrow_n \nu$ if there
are $\eta_0, \eta_1, \dots, \eta_n \in \DistrF$ such that
$\mu = \eta_0 \tauhidearrow \eta_1 \tauhidearrow \dots \tauhidearrow
\eta_n = \nu$.  Clearly, it holds that $\mu \Rightarrow_n \nu$ for
some $n \in \IN$ in case $\mu \Rightarrow \nu$, because $\Rightarrow$
is the transitive closure of $\tauhidearrow$.

We have the following pendant of Lemma~\ref{partial transitions closed
  3} for~$\Rightarrow_n$.

\begin{lemma}
  \label{bounded weak transitions closed}
  Let
  $(\mu_i)_{i{=}0}^\infty, (\nu_i)_{i{=}0}^\infty \in
  \Distr(\calF)^\infty$, $\lim_{\, i \rightarrow\infty} \mu_i = \mu$
  and $\lim_{\, i \rightarrow\infty} \nu_i = \nu$.  If
  $\mu_i \Rightarrow_n \nu_i$ for all $i \in \IN$ then
  $\mu \Rightarrow_n \nu$.
\end{lemma}

\begin{proof}
  By induction on $n$. Basis, $n=0$: Trivial.
  Induction step, $n{+}1$: Given
  $(\mu_i)_{i{=}0}^\infty, (\nu_i)_{i{=}0}^\infty \in
  \Distr(\calF)^\infty\!$, $\mu = \lim_{\, i \rightarrow\infty} \mu_i$,
  and $\nu = \lim_{\, i \rightarrow\infty} \nu_i$, suppose
  $\mu_i \Rightarrow_{n+1} \nu_i$ for all $i \in \IN$.
  \vspace{2pt}
  Let $(\eta_i)_{i{=}0}^\infty \in \Distr(\calF)^\infty$ be such that
  $\mu_i \tauhidearrow \eta_i \Rightarrow_n \nu_i$ for all $i \in
  \IN$.
  Since $\DistrF$ is \vspace{2pt} sequentially compact, the
  sequence~$(\eta_i)_{i{=}0}^\infty$ has a \vspace{2pt} convergent
  subsequence $(\eta_{i_k})_{k{=}0}^\infty$; put
  $\eta = \lim_{\, k \rightarrow\infty}\eta_{i_k}$.
  Because $\mu_{i_k} \tauhidearrow \eta_{i_k}$ for all $k \in \IN$,
  one has $\mu \tauhidearrow\eta$ by Lemma~\ref{partial transitions
    closed 3}.
  Since $\eta_{i_k} \Rightarrow_n \nu_{i_k}$ for $k \in \IN$, the
  induction hypothesis yields $\eta \Rightarrow_n \nu$.
  It follows that $\mu \Rightarrow_{n+1}\nu$.
\end{proof}

\noindent
We adapt Lemma~\ref{bounded weak transitions closed} to obtain a
continuity result for weak transitions~$\,\Arrow$.
  
\begin{lemma}
  \label{weak transitions closed}
  Let
  $(\mu_i)_{i{=}0}^\infty, (\nu_i)_{i{=}0}^\infty \in
  \Distr(\calF)^\infty$, $\lim_{\, i \rightarrow\infty} \mu_i = \mu$
  and $\lim_{\, i \rightarrow\infty} \nu_i = \nu$.  If
  $\mu_i \Rightarrow \nu_i$ for all $i \in \IN$, then
  $\mu \Rightarrow \nu$.
\end{lemma}

\begin{proof}
  Since $\calF$ contains only finitely many non-deterministic
  processes, which can do finitely many $\tau$-transitions only,
  a global upperbound~$N$ exists such that if $\mu \Rightarrow \nu$
  then $\mu \Rightarrow_k \nu$ for some $k \leqslant N$.

  Moreover, as each sequence
  $\mu = \eta_0 \tauhidearrow \eta_1 \tauhidearrow \dots \tauhidearrow
  \eta_k = \nu$ with $k < N$ can be extended to a sequence
  $\mu = \eta_0 \tauhidearrow \eta_1 \tauhidearrow \dots \tauhidearrow
  \eta_{N} = \nu$, namely by taking $\eta_i = \nu$ for all
  $k < i \leqslant N$, on~$\calF$ the relations $\Rightarrow$
  and~$\Rightarrow_N$ coincide.  Consequently, Lemma~\ref{weak
    transitions closed} follows from Lemma~\ref{bounded weak
    transitions closed}.
\end{proof}

\noindent
The following theorem says that equivalence classes of branching
probabilistic bisimilarity in $\DistrF$ are closed sets of
distributions.

\begin{theorem}
  \label{bisimilarity compact}
  Let $\muhat, \nuhat \in \Distr(\calF)$ and
  $(\nu_i)_{i{=}0}^\infty \in \DistrF^\infty$ such that
  $\muhat \brbisim \nu_i$ for all $i \in \IN$ and
  $\nuhat = \lim_{\, i \rightarrow\infty} \nu_i$. Then it holds
  that~$\muhat \brbisim \nuhat$.
\end{theorem}

\begin{proof}
  Define the relation $\calR$ on~$\DistrF$ by
  \begin{displaymath}
    \mu \mopcalR \nu \iff
    \begin{array}[t]{l}
      \exists (\mu_i)_{i{=}0}^\infty, (\nu_i)_{i{=}0}^\infty
      \in \Distr(\calF)^\infty \colon \smallskip \\ \qquad
      \lim_{\, i \rightarrow\infty} \mu_i = \mu \land
      \lim_{\, i \rightarrow\infty} \nu_i = \nu \land
      \forall i \mathbin\in \IN \colon \mu_i\brbisim \nu_i
    \end{array}
  \end{displaymath}
  As $\muhat \mopcalR \nuhat$ (taking $\mu_i:=\muhat$ for all $i \in I$),
  it suffices to show that $\calR$~is a branching probabilistic bisimulation.

  Suppose $\mu \mopcalR \nu$. Let
  $(\mu_i)_{i{=}0}^\infty,(\nu_i)_{i{=}0}^\infty \mathbin\in
  \Distr(\calF)^\infty$ be such that
  $\lim_{\, i \rightarrow\infty} \mu_i\mathbin = \mu$,
  $\lim_{\, i \rightarrow\infty} \nu_i\mathbin = \nu$, and
  $\mu_i\brbisim \nu_i$ for all $i \mathbin\in \IN$.
  Since $\lim_{\, i \rightarrow\infty} \mu_i\mathbin = \mu$, there exist
  $(\mu'_i)_{i{=}0}^\infty \mathbin \in \Distr(\calF)^\infty$ and
  $(r_i)_{i{=}0}^\infty\in\IR^\infty$ such that
  $\mu_i = (1-r_i) \mkern1mu \mu \mathbin\oplus r_i \mkern1mu \mu'_i$
  for all $i \in \IN$ and
  $\lim_{\, i \rightarrow\infty} \mkern1mu r_i = 0$.

  (i)~Towards weak decomposability of~$\calR$ for $\mu$ vs.~$\nu$,
  suppose $\mu = \bigoplusjinJ \: q_j \stardot \mubar_j$.  So, for all
  $i \mathbin\in \IN$, we have that
  $\mu_i = (1 - r_i) \mkern1mu \bigl( \bigoplusjinJ \: q_j \stardot
  \mubar_j \bigr) \mathrel\oplus r_i \mkern1mu \mu'_i$. By weak decomposability
  of~$\brbisim$, there exist $\bar{\nubar}_i$, $\nu'_i$ and $\nu_{ij}$
  for~$i \mathbin\in \IN$ and~$j \mathbin\in J$ such that
  $\nu_i \Arrow{} \bar{\nubar}_i$, $\mu_i \brbisim \bar{\nubar}_i$,
  $\bar{\nubar}_i = (1-r_i) \bigl( \bigoplusjinJ \: q_j \stardot \nu_{ij}
  \bigr) \mathrel\oplus r_i \mkern1mu \nu'_i$,
  $\mu'_i \brbisim \nu'_i$, and $\mubar_j \brbisim \nu_{ij}$
  for~$j \mathbin\in J$.

  The sequences $( \nu_{ij} )_{i{=}0}^\infty$ for~$j \in J$ may not
  converge. However, by sequential compactness of~$\DistrF$ (and
  successive sifting out for each~$j \in J$) an index sequence
  $( i_k)_{k{=}0}^\infty$ exists such that the sequences
  $( \nu_{i_kj} )_{k{=}0}^\infty$ converge, say
  $\lim_{\, k \rightarrow \infty} \nu_{i_kj} = \nubar_j$
  for~$j \in J$. Put
  $\nubar = \bigoplusjinJ \: q_j \stardot \nubar_{j}$.
  Then it holds that
  \begin{displaymath}
    \lim_{\, k \rightarrow \infty} \bar\nubar_{i_k} =
    \lim_{\, k \rightarrow \infty}
    (1-r_{i_k}) \bigl( \bigoplusjinJ \: q_j \stardot \nu_{i_kj} \bigr)
    \mathrel\oplus r_{i_k} \mkern1mu \nu'_{i_k} =
    \lim_{\, k \rightarrow \infty}
    \bigoplusjinJ \: q_j \stardot \nu_{i_kj} =
    \bigoplusjinJ \: q_j \stardot \nubar_{j} =
    \nubar
  \end{displaymath}
  as $\lim_{\, k \rightarrow \infty} r_{i_k} = 0$ and probabilistic
  composition is continuous.
  Since $\nu_{i_k} \Arrow{} \bar\nubar_{i_k}$ for all $k \in \IN$, one has
  $\lim_{\, k \rightarrow \infty} \nu_{i_k} \Arrow{} \lim_{\, k
    \rightarrow \infty} \bar\nubar_{i_k}$, i.e.\ $\nu \Arrow{} \nubar$, by
  Lemma~\ref{weak transitions closed}. Also,
  $\mu_{i_k} \brbisim \bar\nubar_{i_k}$ for all $k \in \IN$. Therefore, by
  definition of~$\calR$, we obtain $\mu \mopcalR \nubar$. Since
  $\mubar_j \brbisim \nu_{i_k j}$ for all $k \in \IN$ and $j \in J$,
  it follows that $\mubar_j \mopcalR \nubar_j$ for~$j \in J$.
  Thus, $\nu \Arrow{}
  \nubar = \bigoplusjinJ \: q_j \stardot \nubar_{j}$, $\mu \mopcalR
  \nubar$, and $\mubar_j \mopcalR \nubar_j$ for all~$j \in J$, as was
  to be shown.
  Hence the relation $\calR$ is weakly decomposable.

  (ii)~For the transfer property, suppose $\mu \arrow{\alpha} \mu'$
  for some $\alpha\in\calA$. Since, for
  each~$i \mathbin\in \IN$, $\mu_i \brbisim \nu_i$ and
  $\mu_i = (1-r_i) \mkern1mu \mu \mathbin\oplus r_i \mkern1mu \mu'_i$, it follows from weak
  decomposability of~$\brbisim$ that
  distributions $\nubar_i$, $\nu'_i$ and~$\nu''_{i}$ exist such that
  $\nu_i \Arrow{} \nubar_i$, $\mu_i \brbisim \nubar_i$,
  $\nubar_i = (1-r_i) \mkern1mu \nu'_i \mathrel\oplus r_i \mkern1mu
  \nu''_{i}$ and $\mu \brbisim \nu'_{i}$.
  By the transfer property for~$\brbisim$, for each $i \in \IN$ exist
  $\etabar_i, \eta'_i \in \DistrE$ such that
    \begin{displaymath}
      \nu'_{i} \Arrow{} \etabar_i, \ 
      \etabar_i \alphahidearrow \eta_i', \ 
      \mu \brbisim \etabar_i, \ \text{and} \ 
      \mu' \brbisim \eta_i'.
    \end{displaymath}
    We have $\nubar'_i \in \DistrF$ for~$i \in \IN$. Also,
    $\etabar_i, \eta'_i \in \DistrF$ for~$i \in \IN$, since $\calF$ is
    assumed to be transition closed. Therefore, by sequential
    compactness of~$\DistrF$, the sequences
    $(\nubar'_{i})_{i{=}0}^\infty$,
    $(\etabar_{i})_{i{=}0}^\infty$,
    $(\etabar'_{i})_{i{=}0}^\infty$ have converging subsequences
    $(\nubar'_{i_k})_{k{=}0}^\infty$,
    $(\etabar_{i_k})_{k{=}0}^\infty$, and
    $(\etabar'_{i_k})_{k{=}0}^\infty$, respectively. Put
    $\nubar = \lim_{\, k \rightarrow \infty} \nu'_{i_k}$,
    $\etabar = \lim_{\, k \rightarrow \infty} \etabar_{i_k}$, and
    $\eta' = \lim_{\, k \rightarrow \infty} \eta'_{i_k}$.  As
    $\lim_{\, k \rightarrow \infty} r_{i_k} = 0$, one has
    $\lim_{\, k \rightarrow \infty} \nubar_{i_k} = \lim_{\, k \rightarrow
      \infty}\nu'_{i_k} = \nubar$.

    Since $\nu_{i_k} \Arrow{} \nubar_{i_k}$ for~$k \in \IN$, we obtain
    $\lim_{\, k \rightarrow \infty} \nu_{i_k} \Arrow{} \lim_{\, k
      \rightarrow \infty} \nubar_{i_k}$ by Lemma~\ref{weak transitions
      closed}, thus $\nu \Arrow{} \nubar$. Likewise, as
    $\nu'_{i_k} \Arrow{} \etabar_{i_k}$ for all $k \in \IN$,
    one has $\nubar \Arrow{} \etabar$, and therefore
    $\nu \Arrow{} \etabar$.  Furthermore, because
    $\etabar_{i_k} \alphahidearrow \eta'_{i_k}$ for~$k \in \IN$, it
    follows that $\etabar\alphahidearrow\eta'$, now by
    Lemma~\ref{partial transitions closed 3}. From
    $\mu \brbisim \etabar_{i_k}$ for all $k \in \IN$, we obtain
    $\mu \mopcalR \etabar$ by definition of~$\calR$. Finally, 
    $\mu' \brbisim \eta'_{i_k}$ for all~$k \in \IN$ yields
    $\mu' \mopcalR \eta'$. Thus $\nu \Arrow{} \etabar \alphahidearrow
    \eta'$, $\mu \mopcalR \etabar$, and~$\mu' \mopcalR \etabar'$,
    which was to be shown.
\end{proof}

\noindent
The following corollary of Theorem~\ref{bisimilarity compact} will be
used in the next section.

\begin{corollary}
  \label{cor-bisimilar derivatives compact}
  For each $\mu \in \DistrE$, the set
  $T_\mu = \lc \nu \mathbin\in \DistrE \mid \nu \brbisim \mu \land \mu
  \mathbin{\Rightarrow} \nu \rc$ is a sequentially compact set.
\end{corollary}

\begin{proof}
  For $\mu = \bigoplusiinI \: p_i \stardot E_i$, the set of processes
  $\calF = \lc E \in \calE \mid \text{$E$ occurs in~$E_i$ for
    some~$i \in I$} \rc$ is finite and closed under
  transitions. Clearly, $\mu \in \DistrF$. Moreover, $\DistrF$ is a
  sequentially compact subset of~$\DistrE$. Taking $\mu_i = \mu$ for
  all $i \in \IN$ in Lemma~\ref{weak transitions closed} yields that
  $\lc \nu \mid \mu \mathbin{\Arrow{}} \nu \rc$ is a closed subset
  of~$\DistrF$. Similarly, the set
  $\lc \nu \mid \nu \brbisim \mu \rc$ is a closed subset
  of~$\DistrF$ by Theorem~\ref{bisimilarity compact}. The statement
  then follows since the intersection of two closed
  subsets of~$\DistrF$ is itself closed, and hence sequentially compact.
\end{proof}

%% file: cancellativity.tex

\section{Cancellativity for branching probabilistic bisimilarity}
\label{sec-cancellativity}

With the results of Section~\ref{sec-continuity} in place, we turn to
stable processes and cancellativity.
In the introduction we argued that in general it doesn't need to be
the case that two branching probabilistic bisimilar distributions
assign the same weight to equivalence classes.
Here we show that this property does hold when restricting to stable
distributions.
We continue to prove the announced unfolding result, that for every
distribution~$\mu$ there exists a stable distribution~$\sigma$ such that
$\mu \Rightarrow \sigma$ and~$\mu \brbisim \sigma$.
That result will be pivotal in the proof of the cancellation theorem,
Theorem~\ref{cancellative}.


\begin{definition}
  \label{def-stable-distr}
  A distribution $\mu \in \DistrE$ is called \emph{stable} if, for all
  $\mubar \in \DistrE$, $\mu \Arrow{} \mubar$ and
  $\mu \mathrel{\brbisim} \mubar$ imply that $\mubar = \mu$.
\end{definition}

\noindent
Thus, a distribution~$\mu$ is called stable if it cannot perform
internal activity \vspace{1pt} without leaving its branching
bisimulation equivalence class.
By definition of $\tauhidearrow$ it is immediate that if
$\bigoplusiinI \: p_i \mydot \mu_i$ is a stable distribution with
$p_i > 0$ for~$i \in I$, then also each probabilistic
component~$\mu_i$ is stable.
Also, because two stable distributions $\mu$ and~$\nu$ don't have any
non-trivial partial $\tau$-transitions, weak decomposability between
them amounts to decomposability, i.e.\ if $\mu \brbisim \nu$ and
$\mu = \bigoplusiinI \: p_i \mu_i$ then distributions~$\nu_i$
for~$i \in I$ exist such that $\nu = \bigoplusiinI \: p_i \nu_i$ and
$\mu_i \brbisim \nu_i$ for~$i \in I$.

\blankline


\noindent
The next result states that, contrary to distributions in general, two
stable distributions are branching bisimilar precisely when they
assign the same probability on all branching bisimilarity classes
of~$\calE$.

\begin{lemma}
  \label{decomp}
  Let $\mu, \nu \in \DistrE$ be two stable distributions. Then it
  holds that $\mu \brbisim \nu$ iff $\mu[C] = \nu[C]$ for each
  equivalence class~$C$ of branching probabilistic bisimilarity
  in~$\calE$.
\end{lemma}

\begin{proof} Suppose $\mu = \bigoplusiinI \: p_i \stardot E_i$,
  $\nu = \bigoplusjinJ \: q_j \stardot F_{\mkern-1mu j}$, and $\mu \brbisim
  \nu$. By weak decomposability,
  $\nu \Arrow{} \nubar = \bigoplusiinI \: p_i \mydot \nu_i$ for
  suitable $\nu_i \in \DistrE$ for~$i \in I$ with
  $\nu_i \brbisim \delta(E_i)$ and $\nubar \brbisim \mu$. Hence,
  $\nubar \brbisim \mu \brbisim \nu$. Thus, by stability of~$\nu$, we
  have $\nubar = \nu$. Say,
  $\nu_i = \bigoplusjinJ \: q_{i \mkern-1mu j} \stardot F_{\mkern-1mu j}$ with
  $q_{i \mkern-1mu j} \geqslant 0$, for $i \in I$,~$j \in J$.
  Since $\nu_i \brbisim \delta(E_i)$, we have by weak decomposability,
  $\delta(E_i) \Arrow{} \bigoplusjinJ \: q_{i \mkern-1mu j} \mydot \mu'_{i \mkern-1mu j}$ such
  that $\delta(E_i) \brbisim \bigoplusjinJ \: q_{i \mkern-1mu j} \mydot \mu'_{i \mkern-1mu j}$
  and $\mu'_{i \mkern-1mu j} \brbisim \delta(F_{\mkern-1mu j})$ for suitable
  $\mu'_{i \mkern-1mu j} \in \DistrE$. Since $\mu$ is stable, so
  is~$\delta(E_i)$. Hence
  $\delta(E_i) = \bigoplusjinJ \: q_{i \mkern-1mu j} \mydot \mu'_{i \mkern-1mu j}$,
  $\mu'_{i \mkern-1mu j} = \delta(E_i)$, and $E_i \brbisim F_{\mkern-1mu j}$ if $q_{i \mkern-1mu j} >
  0$. Put $p_{i \mkern-1mu j} = p_i \mkern1mu q_{i \mkern-1mu j}$, $E_{i \mkern-1mu j} = E_i$ if
  $q_{i \mkern-1mu j} > 0$, and $E_{i \mkern-1mu j} = \bfzero$ otherwise, $F_{i \mkern-1mu j} = F_{\mkern-1mu j}$ if
  $q_{i \mkern-1mu j} > 0$, and $F_{i \mkern-1mu j} = \bfzero$ otherwise, for $i \in
  I$,~$j \in J$. Then it holds that
  \begin{displaymath}
    \begin{array}{l@{\;}c@{\;}l@{\;}c@{\;}l@{\;}c@{\;}l@{\;}c@{\;}l}
      \mu
      & =
      & \bigoplusiinI \: p_i \stardot E_i
      & =
      & \bigoplusiinI \: p_i \mydot
        \bigl( \bigoplusjinJ \: q_{i \mkern-1mu j} \stardot E_i \bigr)
      & =
      & \bigoplusiinI \, \bigoplusjinJ \: p_i \mkern1mu q_{i \mkern-1mu j} \stardot E_i
      & =
      & \bigoplusiinI \, \bigoplusjinJ \: p_{i \mkern-1mu j} \stardot E_{i \mkern-1mu j} 
        \smallskip \\
      \nu
      & =
      & \bigoplusiinI \: p_i \mydot \nu_i
      & =
      & \bigoplusiinI \: p_i \mydot
        \bigl( \bigoplusjinJ \: q_{i \mkern-1mu j} \stardot F_{\mkern-1mu j} \bigr) 
      & =
      & \bigoplusiinI \, \bigoplusjinJ \: p_i \mkern1mu q_{i \mkern-1mu j} \stardot F_{\mkern-1mu j}
      & =
      & \bigoplusiinI \, \bigoplusjinJ \: p_{i \mkern-1mu j} \stardot F_{i \mkern-1mu j}
        \mkern1mu .
    \end{array}
  \end{displaymath}
  Now, for any equivalence class~$C$ of~$\calE$ modulo~$\brbisim$, it
  holds that $E_{i \mkern-1mu j} \in C \Leftrightarrow F_{i \mkern-1mu j} \in C$ for all indices
  $i \in I$,~$j \in J$. So,
  $\mu[C] =
  \textstyle{\sum_{i {\in} I, j {\in} J \colon E_{i \mkern-1mu j} \in C}} \: p_{i \mkern-1mu j} = 
  \textstyle{\sum_{i {\in} I, j {\in} J \colon F_{i \mkern-1mu j} \in C}} \: p_{i \mkern-1mu j} = 
  \nu[C]$.

  For the reverse direction, suppose
  $\mu = \bigoplusiinI \: p_i \stardot E_i$,
  $\nu = \bigoplusjinJ \: q_j \stardot F_{\mkern-1mu j}$, with $p_i, q_j > 0$, and
  $\mu[C] = \nu[C]$ for each equivalence class
  $C \in \calE / {\bisim}$.

  For $i \in I$ and $j \in J$, let $C_i$ and~$D_{\!j}$ be the
  equivalence class in~$\calE$ of $E_i$ and~$F_{\mkern-1mu j}$
  modulo~$\brbisim$. Define
  $r_{i \mkern-1mu j} = \delta_{\mkern1mu i \mkern-1mu j} p_i q_j / \mu [ C_i ]$,
  for $i \in I$, $j \in J$, where
  $\delta_{\mkern1mu i \mkern-1mu j} = 1$
  if~$E_i \brbisim F_{\mkern-1mu j}$ and
  $\delta_{\mkern1mu i \mkern-1mu j} = 0$ otherwise. Then it holds
  that
  \begin{displaymath}
    \textstyle{\sumjinJ} \: r_{i \mkern-1mu j} =
    \textstyle{\sumjinJ} \:
    \mysmallfrac{ \delta_{\mkern1mu i \mkern-1mu j} p_i \mkern1mu q_j }{ \mu [ C_i]
      \rule{0pt}{9pt} } = 
    \mysmallfrac{ p_i }{ \mu [ C_i ] \rule{0pt}{9pt} }
    \textstyle{\sumjinJ} \: \delta_{\mkern1mu i \mkern-1mu j}
    \mkern1mu q_j =
    \mysmallfrac{ p_i \nu [ C_i ]}{ \mu [ C_i ] \rule{0pt}{9pt} } =
    p_i.
  \end{displaymath}
  Since
  $\delta_{\mkern1mu i \mkern-1mu j} p_i \mkern1mu q_j / \mu[C_i] =
  \delta_{\mkern1mu i \mkern-1mu j} p_i \mkern1mu q_j / \nu[D_{\!j}]$ for
  $i \in I$, $j \in J$, we also have
  $\sumiinI \: r_{i \mkern-1mu j} = q_j$. Therefore, we can write
  $\mu = \bigoplusiinI \, \bigoplusjinJ \: r_{i \mkern-1mu j} \stardot
  E_{i \mkern-1mu j}$ and
  $\nu = \bigoplusiinI \, \bigoplusjinJ \: r_{i \mkern-1mu j} \stardot
  F_{i \mkern-1mu j}$ for suitable $E_{i \mkern-1mu j}$
  and~$F_{i \mkern-1mu j}$ such that
  $E_{i \mkern-1mu j} \brbisim F_{i \mkern-1mu j}$.  Calling
  Lemma~\ref{lemma-congruence-for-prob-composition} it follows that
  $\mu \brbisim \nu$.
\end{proof}


\noindent
Next, in Lemma~\ref{stabilizing}, we are about to prove a crucial
property for our proof of cancellativity, the proof of
Theorem~\ref{cancellative} below. Generally, a distribution may
allow inert partial transitions. However, the
distribution can be unfolded to reach via inert partial transitions a
stable distribution, which doesn't have these by definition. To obtain
the result we will rely on the topological property of sequential
compactness of the set
$T_\mu = \lc \mu' \mid \mu' \brbisim \mu \land \mu \mathbin{\Arrow{}}
\mu' \rc$ introduced in the previous section.

\begin{lemma}
  \label{stabilizing}
  For all $\mu \in \DistrE$ there is a stable
  distribution~$\sigma \in \DistrE$ such that $\mu \Rightarrow \sigma$.
\end{lemma}

\begin{proof}
  Define the \emph{weight} of a distribution by
  $\wgt(\mu) \mathbin= \sum_{E \in \calE} \: \mu(E)\cdot \mkern1mu c(E)$,
  i.e.,\ the weighted\vspace{1pt} average of the complexities of the
  states in its support. In view of these definitions,
  $E \arrow\alpha \mu$ implies $\wgt(\mu) < \wgt(\delta(E))$ and
  $\mu \arrow\alpha \mu'$ implies $\wgt(\mu')< \wgt(\mu)$. In addition,
  $\mu \mathbin{\Arrow{}} \mu'$ implies
  $\wgt(\mu') \mathbin\leqslant \wgt(\mu)$.

  For a distribution $\mu \mathbin\in \DistrE$, the set~$T_\mu$ is
  given by
  $T_\mu \mathrel{=} \lc \mu' \mid \mu' \brbisim \mu \land \mu
  \mathbin{\Arrow{}} \mu' \rc$. Consider the value
  $\inf \lc \wgt(\mu') \mid {\mu' \in T_\mu} \rc$.  By
  Corollary~\ref{cor-bisimilar derivatives compact}, $T_\mu$ is a
  sequentially compact set. Since the infimum over a sequentially
  compact set will be reached, there exists a distribution~$\sigma$
  such that $\mu \Arrow{} \sigma$, $\sigma \brbisim \mu$, and
  $\wgt(\sigma) = \inf \lc \wgt(\mu') \mid {\mu' \in T_\mu} \rc$. By
  definition of~$T_\mu$, the distribution~$\sigma$ must be stable.
\end{proof}


\noindent
We have arrived at the main result of the paper, slightly more general
formulated compared to the description in the introduction. The
message remains the same: if two distributions are branching
probabilistic bisimilar and have components that are branching
probabilistic bisimilar, then the components that remain after
cancelling the earlier components are also branching probabilistic
bisimilar. As we see, the previous lemma is essential in the proof as
given.

\begin{theorem}[Cancellativity]
  \label{cancellative}
  Let $\mu, \mu', \nu, \nu' \in \DistrE$ and $0 < r \leqslant 1$ be
  such that $\mu \oplusr \nu \brbisim \mu' \oplusr \nu'$ and
  $\nu \brbisim \nu'$. Then it holds that $\mu \brbisim \mu'$.
\end{theorem}

\begin{proof}
  Choose $\mu$, $\mu'$, $\nu$, $\nu'$, and~$r$ according to the
  premise of the theorem. By Lemma~\ref{stabilizing}, a stable
  distribution~$\sigma$ exists such that
  $\mu \oplusr \nu \Arrow{} \sigma$ and
  $\sigma \brbisim \mu \oplusr \nu$. By weak decomposability, we can
  find distributions $\mubar$ and~$\nubar$ such that
  $\sigma \Arrow{} \mubar \oplusr \nubar$, $\mubar \brbisim \mu$, and
  $\nubar \brbisim \nu$. By stability of~$\sigma$ we have
  $\sigma = \mubar \oplusr \nubar$. Thus $\mubar \oplusr \nubar$ is
  stable. Symmetrically, there are distributions $\mubar'$
  and~$\nubar'$ such that $\mubar' \brbisim \mu'$,
  $\nubar' \brbisim \nu'$ and such that $\mubar' \oplusr \nubar'$ is
  stable. Note,
  $\mubar \oplusr \nubar \brbisim \mu \oplusr \nu \brbisim \mu'
  \oplusr \nu' \brbisim \mubar' \oplusr \nubar'$.

  Let $C \subseteq \calE$ be an equivalence class
  of~$\calE / {\brbisim}$. The distributions $\mubar \oplusr \nubar$
  and $\mubar' \oplusr \nubar'$ are stable and
  ${\mubar \oplusr \nubar} \brbisim {\mubar' \oplusr \nubar'}$. From
  Lemma~\ref{decomp} we obtain that
  $(\mubar \oplusr \nubar)[C] = (\mubar' \oplusr\nubar')[C]$.  Since
  $\nu$ and~$\nubar$ are stable and $\nubar \brbisim \nubar'$, we have
  $\nubar[C] = \nubar'[C]$ for the same reason. Because
  $(\mubar \oplusr \nubar)[C] = r \cdot \mubar[C] + (1{-}r) \cdot
  \nubar[C]$ and
  $(\mubar' \oplusr \nubar')[C] = r \cdot \mubar'[C] + (1{-}r) \cdot
  \nubar'[C]$, we calculate
  \begin{displaymath}
    r \cdot \mubar[C]
    = 
    (\mubar \oplusr \nubar)[C] - (1{-}r) \cdot \nubar[C]
    =
    (\mubar' \oplusr \nubar')[C] - (1{-}r) \cdot \nubar'[C]
    = 
    r \cdot \mubar'[C] \mkern1mu.
  \end{displaymath}
  Since $r \neq 0$, it follows $\mubar[C] = \mubar'[C]$. Since
  $\mubar$ and~$\mubar'$ are stable it follows by Lemma~\ref{decomp}
  that $\mubar \brbisim \mubar'$. Consequently,
  $\mu \brbisim \mubar \brbisim \mubar' \brbisim \mu'$. In particular
  $\mu \brbisim \mu'$, as was to be shown.
\end{proof}

%% file: conclusion.tex

\section{Concluding remarks}
\label{sec-conclusion}

We have shown a cancellation law for distributions with
respect to branching probabilistic bisimilarity.
The result rests on the notion of a stable distribution.
Stable distributions enjoy two properties that have been essential to
our set-up.
(i)~Every distribution has a weak unfolding towards a stable
distribution that is branching probabilistic bisimilar.
(ii)~Branching probabilistic
bisimilarity for stable distributions is determined by their summed
probability for equivalence classes of non-deterministic processes.
Techniques from metric topology have been used to establish the first
result.

We used the cancellativity result in \cite{GGV19:fc} in
order to obtain a complete axiomatisation of branching probabilistic
bisimilarity.
The technical report~\cite{GGV19:tue} contains a proof
sketch in line with this paper.
Yet, as cancellativity is such a fundamental property, and the notion
of branching probabilistic bisimulation is mathematically quite
involved, we regard it necessary to provide a full, detailed proof.

The continuity results of Section~\ref{sec-continuity}, as well as the
argumentation from metric topology at other places, are exploited to
deal with the uncountable number of inert transitions that arise from
combined transitions.
One may wonder if the main theorems of the paper can be achieved based
on combinatorial arguments.
Intuitively, transitions span a convex polyhedron and the
uncountability of the branching of transitions may be reduced to the
finiteness of the transitions spanning the polyhedron.
Despite a number of attempts, we have been forced to leave
the question of a simpler combinatorial proof open.

We leave it as open question for future research weather
cancellativity holds for larger classes of probabilistic processes, as
could be obtained, for instance, by adding recursion, uncountable
choice and/or parallel composition to the syntax.
A further topic for
future research is the study of cancellativity for other weak variants
of probabilistic bisimulation, in particular weak probabilistic
bisimulation.

Other future work is to be devoted
to the construction of an efficient decision algorithm for branching
probabilistic bisimilarity.
A decision procedure for strong probabilistic bisimilarity based on
so-called extended ordered binary trees has been proposed
in~\cite{BEM00:jcss}. An improved algorithm based on partition
refinement is presented in~\cite{GRV18:alg}.
Partition refinement algorithms for weak and branching probabilistic
bisimilarity on states are proposed in~\cite{TH15:ic}.
Reduction of weak probabilistic bisimilarity checking of the
state-based approach of~\cite{CS02:concur} to linear programming is
studied in~\cite{FHHT16:facs}.
Although it is currently not clear how to construct an algorithm
deciding branching probabilistic bisimilarity as put forward in this
paper, it is likely that the procedures of~\cite{GV90:icalp} and \cite{TH15:ic}
can serve as a starting point.